\documentclass[11pt]{article}
\usepackage{outlines}
\usepackage{graphicx}
\usepackage{amssymb}
\usepackage{multirow}
\usepackage{subfigure}
\usepackage{amsmath, thm-restate}
\usepackage{amsthm}  
\usepackage{algorithmic}
\usepackage[section]{algorithm}
\usepackage{subfigure}
\usepackage{color}
\usepackage[normalem]{ulem}
\usepackage{multirow}
\usepackage{fullpage}
\newskip\subfigcapskip	\subfigcapskip	= 2ex

\usepackage{url}

\begin{document}

\author{Chao Li\\
University of Massachusetts \\Amherst, MA, USA\\ chaoli@cs.umass.edu
\and Gerome Miklau\\
\begin{tabular}{cc}
University of Massachusetts & INRIA\\
Amherst, MA, USA & Saclay, France\\ 
\multicolumn{2}{c}{miklau@cs.umass.edu}
\end{tabular}
}

\setlength{\columnsep}{.7cm}
	
\title{Optimal error of query sets under the differentially-private matrix mechanism}

\maketitle

\pagestyle{plain}


\abovedisplayskip = 3pt
\belowdisplayskip = 3pt
\subfigcapskip=-5pt


\floatname{algorithm}{Program}

\newcommand{\cell}{\phi}

\newcommand{\reals}{R}
\newcommand{\vol}{\textup{Vol}}
\newcommand{\convex}{\textup{Convex}}

\newcommand{\minerror}{\mbox{\sc MinError}}
\newcommand{\minsens}{\mbox{\sc MinSensitivity}}
\newcommand{\allrange}{\mbox{\sc AllRange}}
\newcommand{\allpred}{\mbox{\sc AllPredicate}}

\newcommand{\eqbydef}{\stackrel{\mathrm{def}}{=}}

\newcommand{\vect}[1]{\mathbf{#1}}
\newcommand{\sens}[1]{\Delta_{#1}}
\newcommand{\inv}[1]{{#1}^{-1}}
\newcommand{\ep}[1]{\inv{({#1}^t{#1})}}

\def\alg{\mathcal{K}}  
\def\LM{\mathcal{L}}	
\def\GM{\mathcal{G}}	
\def\MM{\mathcal{M}}	
\def\OM{\mathcal{P}}	

\def\lbl{\mbox{LSA}}
\def\svdb{\mbox{\sc svdb}}
\def\ssvdb{\overline{\mbox{\sc svdb}}}
\def\cols{\mbox{cols}}

\def\tr{\mbox{trace}}
\def\var{\mbox{Var}}
\newcommand{\error}[2]{\mbox{\sc Error}_{#1}( #2 )}
\newcommand{\totalerror}[2]{\mbox{\sc TotalError}_{#1}( #2 )}
\newcommand{\maxerror}[2]{\mbox{\sc MaxError}_{#1}( #2 )}

\def\aa{\mathbb{A}}  
\def\bb{\mathbb{B}}  
\def\WW{\mathcal{W}}
\def\PP{\mathbb{P}}
\def\MP{\mathbb{M}}
\def\Mu{\mathcal{U}}

\def\plus{{\!+}}
\def\b{\vect{\tilde{b}}}  
\def\x{\vect{x}}  
\def\estx{\vect{\hat x}}
\def\y{\vect{y}}
\def\q{\vect{q}}  
\def\w{\vect{w}} 
\def\v{\vect{v}}  
\def\estw{\vect{\hat w}}
\def\estq{\vect{\hat q}}
\def\A{\vect{A}}
\def\B{\vect{B}}
\def\Q{\vect{Q}}
\def\W{\vect{W}}
\def\M{\vect{M}}
\def\D{\vect{D}}
\def\P{\vect{P}}
\def\K{\vect{K}}
\def\p{\vect{p}}
\def\I{\vect{I}}
\def\V{\vect{V}}
\def\H{\vect{H}}
\def\G{\vect{G}}
\def\R{\vect{R}}
\def\X{\vect{X}}
\def\Wav{\vect{Y}}
\def\lambdaB{\vect{\lambda}}
\def\LambdaB{\vect{\Lambda}}

\def\PM{\P_{\M}}
\def\DM{\D_{\M}}
\def\DS{\D_s}
\def\DSinv{\DS^{-1}}


\def\Wbool{\W_{01}}		
\def\Wrang{\W_{R}}		
\def\Wunit{\W_{unit}}		

\def\real{\mathbb{R}}

\def\RR{\vect{R}}

\newcommand{\ff}[1]{#1}
\newcommand{\dif}[1]{\mathbf{\delta}_{#1}}


\newcommand{\mh}[1]{}
\newcommand{\gm}[1]{[[\emph{\color{red}GM: #1}]]}
\newcommand{\eat}[1]{}
\newcommand{\cut}[1]{}

\newcommand{\set}[1]{\{#1\}}   

\newtheorem{definition}{Definition}[section]
\newtheorem{proposition}[definition]{Proposition}
\newtheorem{corollary}[definition]{Corollary}
\newtheorem{conjecture}[definition]{Conjecture}
\newtheorem{property}[definition]{Property}
\newtheorem{theorem}[definition]{Theorem}
\newtheorem{problem}[definition]{Problem}
\newtheorem{example}[definition]{Example}
\newtheorem{remark}[definition]{Remark}
\newtheorem{mylemma}[definition]{Lemma}

\def\nbrs{nbrs}
\def\<{\langle}
\def\>{\rangle}

\def\qq{\tilde{q}}
\def\qbar{\overline{q}}

\def\Q{\mathbf{Q}}
\def\QQ{\mathbf{\tilde{Q}}}
\def\QC{\mathbf{\overline{Q}}}
\def\qq{\tilde{q}}
\def\qbar{\overline{q}}

\def\H{\mathbf{H}}
\def\HH{\mathbf{\tilde{H}}}
\def\HC{\mathbf{\overline{H}}}
\def\hh{\tilde{h}}
\def\hbar{\overline{h}}

\def\Lap{\mbox{Laplace}}
\def\Nor{\mbox{Normal}}
\def\cnt{c}
\def\cons{\gamma}
\def\db{I}

\def\hght{{\ell}}  
\def\hv{\hght(v)}
\def\wt{\alpha}
\def\root{r}

\newcommand{\frob}[1]{||#1||_f}
\newcommand{\E}{\mathbb{E}}
\newcommand{\Ldist}[3]{||#1 -#2||_{#3}}
\newcommand{\rank}{\textup{rank}}
\newcommand{\trace}{\textup{Trace}}

\newcommand{\Ltwo}[1]{||#1||_2}
\newcommand{\Lone}[1]{\left\Vert #1  \right\Vert_1}

\newcommand{\reffull}[1]{#1}

\newtheorem{lemma}{Lemma}

\newcommand{\one}[1]{\mathbb{I}_{#1}}
\def\U{\mathcal U}
\def\Z{succZ} 
\def\s{s}
\def\m{M}
\def\mm{\tilde{M}}


\begin{abstract}

A common goal of privacy research is to release synthetic data that satisfies a formal privacy guarantee and can be used by an analyst in place of the original data.  To achieve reasonable accuracy, a synthetic data set must be tuned to support a specified set of queries accurately, sacrificing fidelity for other queries.   

This work considers methods for producing synthetic data under differential privacy and investigates what makes a set of queries ``easy'' or ``hard'' to answer.  We consider answering sets of linear counting queries using the matrix mechanism \cite{Li:2010Optimizing-Linear}, a recent differentially-private mechanism that can reduce error by adding complex correlated noise adapted to a specified workload.  

Our main result is a novel lower bound on the minimum total error required to simultaneously release answers to a set of workload queries.  The bound reveals that the hardness of a query workload is related to the spectral properties of the workload when it is represented in matrix form.  The bound is most informative for $(\epsilon,\delta)$-differential privacy but also applies to $\epsilon$-differential privacy.


\end{abstract}
\newpage

\section{Introduction}

Differential privacy \cite{Dwork:2006Calibrating-Noise} is a rigorous privacy standard offering participants in a data set the appealing guarantee that released query answers will be nearly indistinguishable whether or not their data is included.  The earliest methods for achieving differential privacy were interactive: an analyst submits a query to the server and receives a noisy query answer.  Further queries may be submitted, but increasing noise will be added and the server may eventually refuse to answer subsequent queries.  

To avoid some of the challenges of the interactive model, differential privacy has often been adapted to a non-inter-active setting where a common goal has been to release a synthetic data set that the analyst can use in place of the original data. There are a number of appealing benefits to releasing a private synthetic database: the analyst need not carefully divide their task into individual queries and can use familiar data processing techniques on the synthetic data; the privacy budget will not be exhausted before the queries of interest have been answered; and data processing can be carried out using the resources of the analyst without revealing tasks to the data owner.

There are limits, however, to private synthetic data generation.  When a synthetic dataset is released, the server no longer controls how many questions the analyst computes from the data.  Dinur and Nissim showed that accurately answering ``too many'' queries of a certain type is incompatible with any reasonable notion of privacy, allowing reconstruction of the database with high probability~\cite{Dinur03Revealing}.  

This tempers the hopes of private synthetic data to some degree, suggesting that if a synthetic dataset is to be private, then it can be accurate only for a specific class of queries and may need to sacrifice accuracy for other queries.  A number of methods have been proposed for releasing accurate synthetic data for specific sets of queries~\cite{Ding:2011fk,Li:2010Optimizing-Linear,Hay:2010Boosting-the-Accuracy,xiao2010differential,xiaodifferentially,blum2008a-learning,barak2007privacy,rastogi2007the-boundary,Yuan12Low-Rank}.  These results show that it is still possible to achieve many of the benefits of synthetic data if the released data is targeted to a {\em workload} of queries that are of interest to the analyst. 


In general, efficient differentially private algorithms for answering sets of queries with minimum error are not known.  The goal of our work is to develop tools that can explain what we informally term the {\em error complexity} of a given workload, which should measure, for fixed privacy parameters, the accuracy with which we can simultaneously answer all queries in the workload.   

Such tools can help us to answer a number of natural questions that arise in the context of private synthetic data generation.  Why is it possible to answer one set of queries more accurately than another?  What properties of the queries, or of their relationship to one another, influence this?  Can lower error be achieved by specializing the query set more closely to the task at hand?  Does the combination of multiple users' workloads severely impact the accuracy possible for the combined workload?


Naive approaches to understanding the ``hardness'' of a query workload are unsatisfying.  For example, one may naturally expect that the greater the number of queries in the workload, the larger the error in simultaneously answering them.  Yet the number of queries in a workload is usually an inadequate measure of its hardness.  Query workload sensitivity \cite{Dwork:2006Calibrating-Noise} is another natural approach. Sensitivity measures the maximum change in all query answers due to an insertion or deletion of a single database record.  Basic differentially private mechanisms (e.g. the Laplace mechanism) add noise to each query in proportion to sensitivity, and in such cases sensitivity does in fact determine error rates.  But better mechanisms can reduce error when answering multiple queries (with no cost to privacy), so that sensitivity alone fails to be a reliable measure.
		

In this paper we seek a better understanding of workload error complexity by reasoning formally about the minimum error achievable for a workload, regardless of the underlying database.  We pursue this goal in the context of a class differentially private algorithms: namely those that are instances of the matrix mechanism (so named because workloads are represented as matrices and analyzed algebraically).  The matrix mechanism can be used to answer sets of linear counting queries, a general class of queries which includes all predicate counting queries, histogram queries, marginals, data cubes, and others. 

The matrix mechanism \cite{Li:2010Optimizing-Linear} exploits the relationships between queries in the workload to construct a complex, correlated noise distribution that offers lower error than standard mechanisms.  It encompasses a range of possible approaches to differentially-private query answering because it must be instantiated with a set of queries, called the ``strategy'', which it uses to derive accurate answers to the workload queries.  This makes the mechanism quite general, since any strategy can be selected.  It includes as a special case a number of recently-proposed techniques for answering various subsets of the class of linear queries including range-count queries, sets of low order marginals, sets of data cubes~\cite{barak2007privacy,xiao2010differential,Hay:2010Boosting-the-Accuracy,Ding:2011fk,chaopvldb12,Yuan12Low-Rank,Yaroslavtsev13Accurate}.  Each of the above techniques can be seen as selecting strategies (either manually or adaptively) that work well for given workloads.  The lower bound presented in this work provides a theoretical method of evaluating the quality of the approaches since it provides a lower bound on the error attained by the best possible strategy.


We use the optimal error achievable under the matrix mechanism as a proxy for workload error complexity.  It is computationally infeasible to compute the optimal strategy for an arbitrary workload, making the assessment of workload complexity a challenge.  Nevertheless, we are able to resolve this challenge through the following contributions:


\begin{itemize}
\item Our main result is a novel lower bound on the minimum total error required to simultaneously release answers to a set of workload queries.  The bound reveals that the ``hardness'' of a query workload is related to the eigenvalues of the workload when it is represented in matrix form.  

\item Under $(\epsilon,\delta)$-differential privacy, we characterize two important classes of workloads for which our lower \\bound is tight. As a consequence, it is possible to directly construct a minimum error mechanism for workloads in these classes.  We also analyze the cases for which our bound is not tight, including when the bound is adapted to $\epsilon$-differential privacy.

\item We compare our lower bound to corresponding bounds on the achieved error of recently-proposed mechanisms \cite{DBLP:conf/stoc/RothR10,hardt2010multiplicative,Gupta:2012uq,DBLP:conf/focs/DworkRV10}.
\end{itemize}

Note that our lower bound on error is a conditional bound: it holds for the class of mechanisms defined by the matrix mechanism,
but not necessarily for all differentially private mechanisms.  Nevertheless, we believe this conditional bound serves as a widely useful tool. First, it helps to resolve a number of open questions about the quality of previously-proposed mechanisms that are instances of the matrix mechanism \cite{Li:2010Optimizing-Linear,xiao2010differential,barak2007privacy,Ding:2011fk,chaopvldb12,Yuan12Low-Rank}.  Second, emerging techniques that can outperform this bound tend to exploit special properties of the input database.  Therefore, the achievable error of those mechanisms no longer reflects only properties of the analysis task (as embodied by the workload) but instead reflects the interaction of the task with the database.  Third, the data-independence of the matrix mechanism makes deployment particularly efficient since the noise distribution is fixed for all input databases once a strategy has been selected.  Our bound therefore helps to clarify the utility possible using data-independent mechanisms and reveals when other methods may be required. 

The organization of the paper is as follows. Section \ref{sec:def} reviews definitions and Section \ref{sec:mechanism} presents the matrix mechanism and properties of workload error.  In Section \ref{sec:svdbound} we present the lower bound and its proof.  Section \ref{sec:svdb:theory} evaluates the tightness and looseness of the bound. We compare our bound with error rates of data-dependent mechanisms in Section~\ref{sec:comparison}.
In Appendix~\ref{sec:operation} we show how the bound interacts with algebraic operations on workloads.  
To aid intuition, we include throughout the paper a series of examples in which we compute our lower bound on workloads of interest and report concrete error rates.  Appendix~\ref{app:proof} includes proofs of results not included in the body of the paper.

\section{Definitions \& Background} \label{sec:def}

In this section we describe our representation of query workloads as matrices, formally define differential privacy, and review linear algebra notation.

\subsection{Data model \& linear queries}

The queries considered in this paper are all counting queries over a single relation.  Let the database $I$ be an instance of a single-relation schema $R(\mathbb{A})$, with attributes $\mathbb{A}=\{A_1, A_2, \dots, \\ A_m\}$.  The crossproduct of the attribute domains, written $dom(\mathbb{A})=dom(A_1) \times \dots \times dom(A_m)$, is the set of all possible tuples that may occur in $I$.  

In order to express our queries, we encode the instance $I$ as a vector $\x$ consisting of {\em cell counts}, each counting the number of tuples in $I$ satisfying a distinct logical cell condition.

\begin{definition}[Cell Conditions]
A cell condition is a Boolean formula which evaluates to True or False on any tuple in $dom(\mathbb{A})$.  A collection of cell conditions $\Phi=\cell_1, \cell_2 \dots \cell_n$ is an ordered list of pairwise unsatisfiable cell conditions: each tuple in $dom(\mathbb{A})$ will satisfy at most one $\cell_i$.  
\end{definition}


The data vector is formed from cell counts corresponding to a collection of cell conditions. 

\begin{definition}[Data vector]
Given instance $I$ and a collection of cell conditions $\Phi = \cell_1, \cell_2 \dots \cell_n$, the data vector $\x$ is the length-$n$ column vector consisting of the non-negative integral counts $x_i = |\{t \in I \:|\: \cell_i(t) \mbox{ is True}\}|$.  
\end{definition}

We may choose to fully represent instance $I$ by defining the vector $\x$ with one cell for every element of $dom(\mathbb{A})$.  Then $\x$ is a bit vector of size $|dom(\mathbb{A})|$ with nonzero counts for each tuple present in $I$.  (This is also a vector representation of the full contingency table built from $I$.)  Alternatively, it may be sufficient to partially represent $I$ by the cell counts in $\x$, for example by focusing on a subset of the attributes of $\mathbb{A}$ that are relevant to a particular workload of interest.  Because workloads are finite, it is always sufficient to consider a finite list of cell conditions, even if attribute domains are infinite.

\begin{example}  Consider the relational schema $R=(name,\\ gradyear, gender, gpa)$ describing students and suppose we wish to form queries over $gradyear$ and $gender$ only, where $dom(gender)=$ $\{M,F\}$ and $dom(gradyear)=\{2011,2012,$ $2013,2014\}$.  Then we can define 8 cell conditions which result from all combinations of $gradyear$ and $gender$.  These are enumerated in Table \ref{tbl:one}(a).
\end{example}

Given a data vector $\x$, queries are expressed as linear combinations of the cell counts in $\x$.  

\begin{definition}[Linear counting query]
A~{\em linear}\\{\em counting query} is a length-$n$ row vector $\q=[q_1 \dots q_n]$ with each $q_i \in \mathbb{R}$.  
The answer to a linear counting query $\q$ on $\x$ is the vector product $\q\x = q_1x_1 + \dots + q_nx_n$.
\end{definition}
If $\q$ consists exclusively of coefficients in $\{0,1\}$, then $\q$ is called a {\em predicate counting query}.  In this case, $\q$ counts the number of tuples in $I$ that satisfy the union of the cell conditions corresponding to the nonzero coefficients in $\q$.


\begin{table*} 
\caption{\label{tbl:one} For schema $R=(name, gradyear, gender, gpa)$, (a) shows 8 cell conditions on attributes $gradyear$ and $gender$.  The database vector $\x$ (not shown) will accordingly consist of 8 counts; (b) shows a sample workload matrix $\W$ consisting of five queries, each described in (c).}
\vspace{3ex}
\centering
\subfigure[Cell conditions $\Phi$]{
\small
\begin{tabular}{l}
$\cell_1: gradyear=2011 \wedge gender=M$\\
$\cell_2: gradyear=2011 \wedge gender=F$ \\
$\cell_3: gradyear=2012 \wedge gender=M$ \\
$\cell_4: gradyear=2012 \wedge gender=F$ \\
$\cell_5: gradyear=2013 \wedge gender=M$ \\
$\cell_6: gradyear=2013 \wedge gender=F$ \\
$\cell_7: gradyear=2014 \wedge gender=M$ \\
$\cell_8: gradyear=2014 \wedge gender=F$ \\
\end{tabular}
}
\quad
\subfigure[A query matrix $\W$]{
\small
$\begin{bmatrix}
1 & 1 & 1 & 1 & 1 & 1 & 1 & 1 \\
1 & 1 & 1 & 1 & 0 & 0 & 0 & 0 \\
0 & 1 & 0 & 1 & 0 & 0 & 0 & 0 \\
1 & 0 & 1 & 0 & 0 & 0 & 0 & 0 \\
0 & 0 & 0 & 0 & 1 & 1 & \mbox{-}1 & \mbox{-}1 \\
\end{bmatrix}$
}
\quad
\subfigure[Counting queries defined by rows of $\W$]{ \small
\begin{tabular}{l}
$\q_1$: all students; \\
$\q_2$: students with $gradyear \in [2011,2012]$;\\
$\q_3$: female students with $gradyear \in [2011,2012]$;\\
$\q_4$: male students with $gradyear \in [2011,2012]$;\\
$\q_5$: difference between 2013 grads and 2014 grads.\\
\end{tabular}
}
\vspace{-5ex}
\end{table*} 

\subsection{Query workloads}
A {\em workload} is a finite set of linear queries.  A workload is represented as a matrix, each row of which is a single linear counting query.  

\begin{definition}[Query matrix]
A {\em query matrix} is a collection of $m$ unique linear counting queries, arranged by rows to form an $m \times n$ matrix.
\end{definition}
Note that cell condition $\cell_i$ defines the meaning of the $i^{th}$ position of $\x$, and accordingly, it determines the meaning of the $i^{th}$ column of $\W$.  Unless otherwise noted, we assume all workloads are defined over the same fixed set of cell conditions.

\begin{example}   
The matrix in Table \ref{tbl:one}(b) shows a workload of five queries.  The first four are predicate queries.  Table \ref{tbl:one}(c) describes the meaning of the queries w.r.t. the cell conditions in Table \ref{tbl:one}(a).
\end{example}

We assume that workloads consist of unique queries, without duplicates.  If workload $\W$ is an $m \times n$ query matrix, the answers for $\W$ are represented as a length $m$ column vector of numerical query results, which can be computed by multiplying matrix $\W$ by the data vector $\x$.  

Note that it is critical that the analyst include in the workload {\em all} queries of interest.   In the absence of noise introduced by the privacy mechanism, it might be reasonable for the analyst to request answers to a small set of counting queries, from which other queries of interest could be computed.  (E.g., it would be sufficient to recover $\x$ itself by choosing the workload defined by the identity matrix.)  But because the analyst will receive private, noisy estimates to the workload queries, the error of queries computed from their combination is often increased.  Our privacy mechanism is designed to optimize error across the entire set of desired queries, so all queries should be included.  

As a concrete example, in Table \ref{tbl:one}(b), $\q_4$ can be computed as $(\q_2 - \q_3)$ but is nevertheless included in the workload.  This reflects the fact that we wish to simultaneously answer all included queries with minimum aggregate error, treating each equally.  It is also possible to scale individual rows by a positive scalar value, which has the effect of reducing the error of that query.

The cell conditions are used to define the semantics of the queries in a workload, while the workload properties we study are primarily determined by features of their matrix representation.  This can lead to a few representational inconsistencies we would like to avoid.  First, while a workload $\W$ is meant to represent a {\em set} of queries, as a matrix it has a specified order of its rows.  Further, for any workload $\W$ defined by cell conditions $\Phi=\cell_1 \dots \cell_n$, consider any permutation $\Phi'$ of $\Phi$.  Then there is a different matrix $\W'$, defined on $\Phi'$ and constructed from $\W$ by applying the permutation to its columns, that is semantically equivalent to $\W$.  We will verify later that our analysis of workloads is representation independent for both rows and columns. We use the following definition:

\begin{definition}[Representation independence] \label{def:rowcol}
A numerical measure $\rho$ on a workload matrix is {\em row (resp. column) representation independent} if, given a workload matrix $\W$, and any workload $\W'$ which results from permuting the rows (columns) of $\W$, $\rho(\W) = \rho(\W')$.
\end{definition}

A related issue arises in the specification of the cell conditions.  If a workload $\W$ is defined by cell conditions in $\Phi$, then we can always consider extending $\Phi$ by adding additional cell conditions not relevant to the queries in $\W$.  We can then define a semantically equivalent workload $\W'$ on $\Phi'$ which will consist of columns of zeroes for each of the new cell conditions.  To address this issue in later sections we rely on the following definition:

\begin{definition}[Column Projection]
Given an $m \times n$ workload $\W$ defined by cell conditions $\Phi=\cell_1 \dots \cell_n$, and an ordered subset $\Psi \subseteq \Phi$ consisting of $p$ selected cell conditions, the column projection of $\W$ w.r.t $\Psi$ is a new $m \times p$ workload consisting only of the columns of $\W$ included in $\Psi$.
\end{definition}
In the rest of the paper, $\mu$ denotes a subset of cell conditions, $|\mu|$ denotes its cardinality and $\Mu_n$ to denote all possible subsets of $n$ cell conditions. $\mu(\W)$ is the column projection of $\W$ w.r.t. $\mu$.
\begin{example}
The column projection of the workload in Table~\ref{tbl:one}(b) w.r.t. cell conditions $\{\phi_1, \phi_3, \phi_5, \phi_7\}$, which consists of queries only over {\em Male} students, is the following:
\[\begin{bmatrix}
1 & 1 & 1 & 1 \\
1 & 1 & 0 & 0 \\
0 & 0 & 0 & 0 \\
1 & 1 & 0 & 0 \\
0 & 0 & 1 & \mbox{-}1 \\
\end{bmatrix}
\]
\end{example}
\paragraph*{Common workloads}  We introduce notation for two common workloads that contain all queries of a certain type. $\allrange(d)$ denotes the workload consisting of all range-count queries over $d$ cell conditions (typically derived from a single ordered attribute $A_i$ with $|dom(A_i)|=d$).  There are $\frac{d}{2}(d+1)$ queries in workload $\allrange(d)$.  Over $k$ ordered attributes with domain sizes $d_1 \dots d_k$, we similarly define the set of all k-dimensional range-count queries, denoted $\allrange(d_1, \dots d_k)$.

{\sc All\-Pred\-icate}$(d)$ is the much larger workload consisting of all predicate counting queries over $d$ cell conditions.  There are $2^d$ queries in $\allpred(d)$.

\begin{example}
The workload of queries counting the  students who have graduated in any interval of years drawn from $\{2011,2012,2013,2014\}$ is denoted $\allrange(4)$ and consists of 10 one-dimensional range queries over the {\em gradyear} attribute.
\end{example}

\begin{example}
\hspace*{-1ex}The set of all two dimensional range-count queries over {\em gradyear} and {\em gender} is written $\allrange(4,2)$, where the possible ``ranges'' for $gender$ are simply $M$, $F$, or $(M \vee F)$.  This workload consists of $30$ queries.  The workload of all predicate queries over over {\em gradyear} and {\em gender} is $\allpred(8)$, consisting of all $256$ predicate queries over a domain of size 8.
\end{example}

\subsection{Differential privacy \& basic mechanisms}

Standard $\epsilon$-differential privacy \cite{Dwork:2006Calibrating-Noise} places a bound (controlled by $\epsilon$) on the difference in the probability of query answers for any two {\em neighboring} databases.  For database instance $\db$, we denote by $\nbrs(\db)$ the set of databases formed by adding or removing exactly one tuple from $\db$. 
Approximate differential privacy~\cite{Dwork:2006Our-Data-Ourselves:,McSherry:2009fk}, is a relaxation in which the $\epsilon$ bound on query answer probabilities may be violated with small probability, controlled by $\delta$.

\begin{definition}[Differential Privacy] 
\hspace*{-1ex}A randomized algorithm $\alg$ is $(\epsilon,\delta)$-differentially private if for any instance $I$, any $I' \in \nbrs(I)$, and any subset of outputs $S \subseteq Range(\alg)$, the following holds:
\[
Pr[ \alg(I) \in S] \leq \exp(\epsilon) \times Pr[ \alg(I') \in S] +\delta
\]		
where the probability is taken over the randomness of the $\alg$.
	\end{definition}
When $\delta=0$, this definition describes standard $\epsilon$-differential privacy.

Both definitions can be satisfied by adding random noise to query answers.  The magnitude of the required noise is determined by the {\em sensitivity} of a set of queries: the maximum change in a vector of query answers over any two neighboring databases.  However, the two privacy definitions differ in the measurement of sensitivity and in their noise distributions.  Standard differential privacy can be achieved by adding Laplace noise calibrated to the $L_1$ sensitivity of the queries \cite{Dwork:2006Calibrating-Noise}. Approximate differential privacy can be achieved by adding Gaussian noise calibrated to the $L_2$ sensitivity of the queries \cite{Dwork:2006Our-Data-Ourselves:,McSherry:2009fk}.  This small difference in the sensitivity metric---from $L_1$ to $L_2$---has important consequences for the theory underlying our analysis and, unless otherwise noted, {\em stated results apply only to approximate differential privacy}.  Sec \ref{sec:svdb:l1} contains a comparison of these two definitions as they pertain to the matrix mechanism and the results of this paper.


Since our query workloads are represented as matrices, we express the sensitivity of a workload as a matrix norm.  Notice that for neighboring databases $\db$ and $\db'$, $|(\db-\db')\cap(\db'-\db)|=1$ and recall that all cell conditions are mutually unsatisfiable.  It follows that the corresponding data vectors $\x$ and $\x'$ differ in exactly one component, by exactly one.  We extend our notation and write  $\x' \in \nbrs(\x)$.  The $L_2$ sensitivity of $\W$ is equal to the maximum $L_2$ norm of the columns of $\W$. Below, $\cols(\W)$ is the set of column vectors $W_i$ of $\W$.

\begin{definition}[$L_2$ Query matrix sensitivity] \label{def:l2sens}
$\quad $ The $L_2$ sensitivity of a query matrix $\W$ is denoted $\sens{\W}$ and defined as follows:
\begin{eqnarray*}
\sens{\W} & \eqbydef & \max_{\x' \in \nbrs(\x)} \Ltwo{\W\x - \W\x'} 
			= \max_{W_i \in \cols(\W)} \Ltwo{W_i}
\end{eqnarray*}
\end{definition}

The classic differentially private mechanism adds independent noise calibrated to the sensitivity of a query workload.  We use $\Nor(\sigma)^m$ to denote a column vector consisting of $m$ independent samples drawn from a Gaussian distribution with mean $0$ and scale $\sigma$.

\begin{proposition}{\sc (Gaussian mechanism \cite{Dwork:2006Our-Data-Ourselves:, McSherry:2009fk})}\label{prop:l2diffpriv}
\hspace*{-0.5ex}Given an $m \times n$ query matrix $\W$, the randomized algorithm $\GM$ that outputs the following vector is $(\epsilon,\delta)$-differentially private:
$$\GM(\W,\x) = \W\x + \Nor(\sigma)^m$$ 
where $\sigma=\sens{\W}\sqrt{2\ln(2/\delta)}/\epsilon$
\end{proposition}

Recall that $\W\x$ is a vector of the true answers to each query in $\W$.  The algorithm above adds independent Gaussian noise (scaled by the sensitivity of $\W$, $\epsilon$, and $\delta$) to each query answer.  Thus $\GM(\W,\x)$ is a length-$m$ column vector containing a noisy answer for each linear query in $\W$.

\subsection{Linear algebra notation}

Throughout the paper, we use the notation of linear algebra and employ standard techniques of matrix analysis.  Recall that for a matrix $\A$, $\A^T$ is its transpose, $\inv{\A}$ is its inverse, and $\tr(\A)$ is the sum of values on the main diagonal.  The Frobenius norm of $\A$ is denoted $||\A||_F$ and defined as the square root of the squared sum of all entries in $\A$, or, equivalently, $\sqrt{\tr(\A^T\A)}$.  We use $diag(c_1, \dots c_n)$ to indicate an $n \times n$ diagonal matrix with scalars $c_i$ on the diagonal.  We use $\vect{0}^{m \times n}$ to indicate a matrix of zeroes with $m$ rows and $n$ columns.  An orthogonal matrix $\Q$ is a square matrix whose rows and columns are orthogonal unit vectors, and for which $\Q^T=\inv{\Q}$.

We will also rely on the notion of a positive semidefinite matrix. A symmetric square matrix $\A$ is called positive semidefinite if for any vector $\x$, $\x^T\A\x\geq 0$.  If $\A$ is positive semidefinite, denoted $\A\succeq 0$, all its diagonal entries are non-negative as well. In particular, for any matrix $\A$, $\A^T\A$ is a positive semidefinite matrix.

In addition, we use $\A^\plus$ to represent the Moore-Penrose pseudoinverse of a matrix $\A$, a generalization of the matrix inverse defined as follows:
\begin{definition}{\sc (Moore-Penrose Pseudoinverse \cite{ben2003generalized})}
Given a $m\times n$ matrix $\A$, a matrix $\A^\plus$ is the Moore-Penrose pseudoinverse of $\A$ if it satisfies each of the following:
\begin{align*}
\A\A^\plus\A&=\A, &\A^\plus\A\A^{\plus}=\A^{\plus},\\
(\A\A^\plus)^T&=\A\A^\plus, & (\A^\plus\A)^T=\A^\plus\A.
\end{align*}
\end{definition}
We include some important properties of the Moore-Penrose pseudoinverse in the following theorem.
\begin{theorem}{\sc (\cite{ben2003generalized})}\label{thm:mpinverse}
The  Moore-Penrose pseudoinverse satisfies the following properties:
\begin{enumerate}\itemsep 0in
\item Given any matrix $\A$, there exists a unique matrix that is the Moore-Penrose pseudoinverse of $\A$. 
\item Given a vector $\y$, we have $||\y-\A\x||_2\geq||\y-\A\A^\plus\y||_2$ for any vector $\x$.
\item For any satisfiable linear system $\B\A=\W$, $\W\A^\plus$ is a solution to the linear system and $||\W\A^\plus||_F\leq ||\B||_F$ for any solution $\B$ to the linear system.
\end{enumerate}
\end{theorem}

Throughout the paper, we use the singular value decomposition of a matrix, which is a classic tool of matrix analysis. If $\W$ is an $m \times n$ matrix, the singular value decomposition (SVD) of $\W$ is a factorization of the form $\W = \Q_\W \LambdaB_\W \P_\W^T$ such that $\Q_\W$ is an $m \times m$ orthogonal matrix, $\LambdaB_\W$ is a $m \times n$ diagonal matrix containing the singular values of $\W$ and $\P_\W$ is an $n \times n$ orthogonal matrix.  When $m > n$, the diagonal matrix $\LambdaB_\W$ consists of an $n \times n$ diagonal submatrix combined with $\vect{0}^{(m-n) \times n}$. In addition, we also consider the eigenvalue decomposition of matrix $\W^T\W$ and the square root of $\W^T\W$. The eigenvalue decomposition of $\W^T\W$ has the form $\W^T\W=\P_\W\D_\W\P_\W^T$, where $\P_\W$ is the same matrix as the singular value decomposition of $\W$ and $\D'_\W$ is an $n \times n$ diagonal matrix such that $\D_\W=\LambdaB_\W^T\LambdaB_\W$. The square root of $\W^T\W$, denoted as $\sqrt{\W^T\W}$, is a matrix $\W'$ such that $(\W')^2 = \W^T\W$, which can also be represented as the singular values and singular vectors of $\W$: $\W'= \P_\W \LambdaB_\W \P_\W^T$.
\section{The $(\epsilon,\delta)$-Matrix Mechanism}\label{sec:mechanism}

In this section we define the class of algorithms that can be constructed using the matrix mechanism, we define optimal error of a workload with respect to this class, and we develop notions of workload equivalence and containment consistent with our error measures.

\subsection{The extended matrix mechanism}

The matrix mechanism~\cite{Li:2010Optimizing-Linear} has a form similar to the Gaussian mechanism in Prop.~\ref{prop:l2diffpriv}, but adds a more complex noise vector.  It uses a different set of queries (the strategy matrix $\A$) to construct this vector. The intuitive justification for this mechanism is that it is equivalent to the following three-step process: (1) the queries in the strategy are submitted to the Gaussian mechanism; (2) an estimate $\estx$ for $\x$ is derived by computing the $\estx$ that minimizes the squared sum of errors (this step consists of standard linear regression and requires that $\A$ be full rank to ensure a unique solution); (3) noisy answers to the workload queries are then computed as $\W\estx$. 

We present an extended version of the matrix mechanism which relaxes the requirement that $\A$ be full rank.  Instead it is sufficient that all queries in $\W$ can be represented as linear combinations of queries in $\A$. Then estimating $\estx$ is not necessary, and step (2) and (3) can be combined.

\begin{restatable}[]{proposition}{extendmmech} \label{def:m-mech} {\hspace*{-0.7ex}\sc(Extended $(\epsilon,\!\delta)$-Matrix Mechanism)}
Given an $m \times n$ query matrix $\W$, and a $p \times n$ strategy matrix $\A$ such that $\W\A^\plus\A=\W$, the following randomized algorithm $\MM_\A$ is $(\epsilon,\delta)$-differentially private:
\begin{eqnarray*}
\MM_\A(\W,\x) &=& \W\x + \W \A^\plus \Nor(\sigma)^m.
\end{eqnarray*}
where $\sigma=\sens{\A}\sqrt{2\ln(2/\delta)}/\epsilon$.
\end{restatable}
The proof is contained in App.~\ref{app:proof:extendmmech}.  

Here condition $\W\A^\plus\A\!=\!\W$ guarantees that the queries in $\A$ can
represent all queries in $\W$. When $\A$ has full rank, $\A^\plus = \inv{(\A^T\A)}\A^T$, which coincides with the original definition~\cite{Li:2010Optimizing-Linear}.  

Like the Gaussian mechanism, the matrix mechanism computes the true answer vector $\W\x$ and adds noise to each component.  But a key difference is that the scale of the Gaussian noise is {\em calibrated to the sensitivity of the strategy matrix $\A$, not that of the workload}.  In addition, the noise added to the list of query answers is no longer independent, because the vector of independent Gaussian samples is transformed by the matrix $\W\A^\plus$.  

The matrix mechanism can reduce error, particularly for large or complex workloads, by avoiding redundancy in the set of desired workload queries.  Intuitively, some workloads consist of queries that ask the same (or similar) questions of the database multiple times, which incurs a significant cost to the privacy budget under the Gaussian mechanism.  By choosing the right strategy matrix for a workload, it is possible to remove the redundancy from the queries submitted to the privacy mechanism and derive more accurate answers to the workload queries. 

Fundamental to the performance of the matrix mechanism, is the choice of the strategy matrix which instantiates it.  One naive approach is to minimize sensitivity.  The full rank strategy matrix with least sensitivity is the identity matrix, $\I$, which has sensitivity 1.  With $\A=\I$, the matrix mechanism privately computes the individual counts in $\x$ and then uses them to estimate any desired workload query.  At the other extreme, the workload itself can be used as the strategy, setting $\A=\W$.  In this case, there is no benefit in sensitivity over the Gaussian mechanism\footnote{Although there is no benefit in sensitivity when $\A=\W$, the matrix mechanism still has lower error than the Gaussian mechanism for some workloads by combining related query answers into a more accurate consistent result.}.  

For many workloads, neither of these basic strategies offer optimal error.  Recent research has shown that for specific workloads, there exist strategies that can offer much better error rates.  For example, if $\W=\allrange(n)$, two strategies were recently proposed. A {\em hierarchical} strategy \cite{Hay:2010Boosting-the-Accuracy} includes the total sum over the whole domain, the count of each half of the domain, and so on, terminating with counts of individual elements of the domain.  The {\em wavelet} strategy \cite{xiao2010differential} consists of the matrix describing the Haar wavelet transformation.  Informally, both strategies achieve low error because they each have low sensitivity, $O(\log n)$, and every range query can be expressed as a linear combination of just a few strategy queries.\footnote{The approaches in \cite{Hay:2010Boosting-the-Accuracy,xiao2010differential} were originally proposed in the context of $\epsilon$-differential privacy, but their behavior is similar under $(\epsilon,\delta)$-differential privacy.}  Although both of these strategy matrices offer significant improvements in error for range query workloads, neither is optimal.  We will use our lower bound to evaluate the quality of these proposed strategies in Sec \ref{sec:svdbound}.  Other recently-proposed techniques can be seen as attempts to compute approximately optimal strategy matrices adapted to the input workload \cite{barak2007privacy,Ding:2011fk,chaopvldb12,Yuan12Low-Rank,Yaroslavtsev13Accurate}.


\subsection{Measuring \& minimizing error} \label{sec:sub:error}

We measure the error of individual query answers using mean squared error.  For a workload of queries, the error is defined as the {\em total} of individual query errors.  

\begin{definition}[Query and Workload Error \cite{Li:2010Optimizing-Linear}]\label{def:wkerror} Let $\estw$ be the estimate for query $\w$ under the matrix mechanism using query strategy $\A$.  That is, $\estw=\MM_\A(\w,\x)$.  The mean squared error of the estimate for $\w$ using strategy $\A$ is: 
\vspace*{-1ex}
$$\error{\A}{\w} \eqbydef \E[ ( \w\x - \estw)^2 ].$$ 
Given a workload $\W$, the total mean squared error of answering $\W$ using strategy $\A$ is: $\error{\A}{\W} =$ \\
$\sum_{\w_i \in \W} \error{\A}{\w_i}$.
\end{definition}

The query answers returned by the matrix mechanism are linear combinations of noisy strategy query answers to which independent Gaussian noise has been added.  Thus, as the following proposition shows (extending the corresponding proposition from \cite{Li:2010Optimizing-Linear}), we can directly compute the error for any linear query $\w$ or workload of queries $\W$:


\begin{restatable}[Total Error]{proposition}{proptotalerror}\label{prop:totalerror}
Given a workload \\$\W$, the total error of answering $\W$ using the extended $(\epsilon,\delta)$-matrix mechanism with query strategy $\A$ is:
\begin{equation}\label{eqn:totalerror}
 \error\A{\W} = P(\epsilon, \delta)\sens{\A}^2 \;||\W\A^\plus||_F^2
 \end{equation}
where $P(\epsilon, \delta)=\frac{2\log(2/\delta)}{\epsilon^2}$.
\end{restatable}
\eat{According to proposition~\ref{prop:totalerror}, when the privacy parameters $\epsilon$ and $\delta$ are chosen, the total error of using strategy $\A$ to answer workload $\W$ is explicitly determined by 
\[|| \A ||_2^2 \tr (\W^T\W(\A^T\A)^{-1}).\]
To make the representation easier, in the rest of this paper, we always assume $P(\epsilon, \delta)=1$ so that
\[ \totalerror\A{\W} = || \A ||_2^2\tr (\W^T\W(\A^T\A)^{-1}).\]}
%
The proof is contained in App.~\ref{app:proof:extendmmech}. 

We call a mechanism \textsl{data-independent} if its error for workload $\W$ is independent of the data vector $\x$. Proposition~\ref{prop:totalerror} shows that the matrix mechanism is data-independent.

The optimal strategy for a workload $\W$ is defined to be the one that minimizes total error (the same measure used in \cite{Li:2010Optimizing-Linear}).  

\begin{definition}[Minimum Total Error \cite{Li:2010Optimizing-Linear}]  \label{def:minerr}
Given a workload $\W$, the minimum total error is:
\begin{equation}\label{eqn:mintotalerror}
	\minerror(\W) = \min_{\A:\, \W\A^\plus\A=\W}\error\A\W.  
\end{equation}
\end{definition}

Our work investigates the error complexity of workloads as it is represented by $\minerror(\W)$.  This reflects the hardness of the workload assuming we use an algorithm that is an instance of the matrix mechanism and that a total squared error measure is used.  Understanding error for this class of mechanisms is a first step towards more general lower bounds and helps to assess the quality of a number of previously-proposed algorithms included in this class~\cite{Li:2010Optimizing-Linear,xiao2010differential,barak2007privacy,Ding:2011fk,chaopvldb12,Yuan12Low-Rank,Yaroslavtsev13Accurate}. 


The strategy matrix that minimizes total error can be computed using a semi-definite program (SDP) \cite{Li:2010Optimizing-Linear}.  However, finding the solutions of the program with standard SDP solvers takes $O(n^8)$ time, where $n$ is the number of cell conditions, making it infeasible for realistic applications.  Efficient approximation algorithms for this problem have been investigated recently~\cite{chaopvldb12,Yuan12Low-Rank}.  Yet, approximate---or even exact---solutions to this problem do not provide much general insight into the main goal of this paper: to understand the properties of workloads that determine the magnitude of $\minerror(\W)$.  The bound we will present in Sec. \ref{sec:svdbound} is important because it closely approximates $\minerror(\W)$, it is easily computable, and it reveals the connection between minimum error and spectral properties. 

Nevertheless, our results do not preclude the existence of different mechanisms capable of answering a workload $\W$ with lower error. In particular, our error analysis does not include data-dependent algorithms~\cite{DBLP:conf/stoc/RothR10, hardt2010multiplicative,nissim2007smooth,Cormode:2012Spatial}.  The error rates for these mechanisms no longer reflect properties of the workload alone, but instead some combination of the workload and properties of the input data, and call for a substantially different analysis.  In our case, $\x$ (i.e., the vector of cell counts corresponding to the database) does not appear in (1) above.  This means that our minimum error strategy depends on the workload alone, independent of a particular database instance.

\subsection{Equivalence \& containment for workloads}\label{sec:mechanism:equiv}

Next we develop a notion of equivalence and containment of workloads with respect to error.  We will verify that the error bounds presented in the next section satisfy these relationships in most cases. 

The special form of the expression for total error in Prop. \ref{prop:totalerror} means that there are many workloads that are equivalent from the standpoint of error. For two workloads $\W_1$ and $\W_2$, if $\W_1^T\W_1 = \W_2^T\W_2$, then any strategy $\A$ that can represent the queries of $\W_1$ can also represent the queries of $\W_2$, and vice versa. In addition, $\W_1^T\W_1 = \W_2^T\W_2$ implies $||\W_1\A^\plus||_F^2=||\W_2\A^\plus||_F^2$ for any strategy $\A$.  We therefore define the following notion of {\em equivalence} of two workloads: 

\begin{definition}[Workload Equivalence] An $m_1 \times n_1$ workload $\W_1$ and an $m_2\times n$ workload $\W_2$ are {\em equivalent}, denoted $\W_1 \equiv \W_2$, if $\W_1^T\W_1 = \W_2^T\W_2$.
\end{definition}

The following conditions on pairs of workloads imply that they have equivalent minimum error: 

\begin{restatable}[Equivalence Conditions]{proposition}{propequiverr}\label{prop:equiverr}
Given\\ an $m_1\times n_1$ workload $\W_1$ and an $m_2\times n_2$ workload $\W_2$, each of the following conditions implies that $\minerror(\W_1)=\minerror(\W_2)$:
\begin{enumerate} \itemsep 0in
\item[(i)] $\W_1 \equiv \W_2$
\item[(ii)] $\W_1=\Q\W_2$ for some orthogonal matrix $\Q$.
\item[(iii)] $\W_2$ results from permuting the rows of $\W_1$.
\item[(iv)] $\W_2$ results from permuting the columns of $\W_1$.
\item[(v)] $\W_2$ results from the column projection of $\W_1$ on all of its nonzero columns.
\end{enumerate}
\end{restatable}
It follows from this proposition that $\minerror$ is row and column representation independent, and behaves well under the projection of extraneous columns.  

Defining a notion of containment for workload matrices is more complex than simple inclusion of rows.  Even if the rows of $\W_1$ are not present in $\W_2$, it could be that $\W_1$ is in fact contained in $\W_2$ when expressed using an alternate basis.  The following definition considers this possibility:

\begin{definition}[Workload Containment]
An $m_1 \times n$ workload $\W_1$ is contained in an $m_2\times n$ workload $\W_2$, denoted $\W_1\subseteq \W_2$, if there exists a $\W_2'\equiv\W_2$ and the rows of $\W_1$ are contained in $\W'_2$.
\end{definition}  

The following proposition shows two conditions which imply inequality of error among workloads:

\begin{restatable}[Error inequality]{proposition}{properrinequal}\label{prop:errorinequal} 
Given an $m_1 \times n_1$ workload $\W_1$ and an $m_2 \times n_2$ workload $\W_2$, each of the following conditions implies that $$\minerror(\W_1) \leq\minerror(\W_2):$$
\begin{enumerate} \itemsep 0in
\item[(i)] $\W_1 \subseteq \W_2$
\item[(ii)] $\W_1$ is a column projection of $\W_2$.
\end{enumerate}
\end{restatable}

App. \ref{app:proof:contain} contains the proof of Prop. \ref{prop:equiverr} and \ref{prop:errorinequal}.

\section{The Singular Value Bound} \label{sec:svdbound}

In this section we state and prove our main result: a lower bound on $\minerror(\W)$, the optimal error of a workload $\W$ under the extended $(\epsilon, \delta)$-matrix mechanism.  The bound shows that the hardness of a workload is a function of its eigenvalues. We describe the measure and its properties in Section~\ref{sec:svdb:svdb} and prove that it is a lower bound in Section~\ref{sec:svdb:proof}.  In Section~\ref{sec:svdb:l1} we briefly discuss the challenge of adapting this bound to $\epsilon$-differential privacy. 


\subsection{The Singular Value Bound}\label{sec:svdb:svdb}

We first present the simplest form of our bound, which is based on computing the square of the sum of eigenvalues of the workload matrix:
\begin{definition}[\sc Singular Value Bound] \label{def:svdb}
Given an $m\times n$ workload $\W$, its singular value bound, denoted $\svdb(\W)$, is:
\[\svdb(\W) = \frac{1}{n}(\lambda_1+\ldots+\lambda_n)^2,\]
where $\lambda_1, \ldots, \lambda_n$ are the singular values of $\W$.
\end{definition}

The following theorem guarantees that the singular value bound is a valid lower bound to the minimal error of a workload. The proof is presented in detail in Sec.~\ref{sec:svdb:proof}.

\begin{theorem}\label{thm:singularvaluebound}
Given an $m\times n$ workload $\W$, 
\[\minerror(\W)\geq P(\epsilon, \delta)\svdb(\W),\]
where $P(\epsilon, \delta)=\frac{2\log(2/\delta)}{\epsilon^2}$.
\end{theorem}

In the rest of paper, we refer to $\svdb(\W)$ as the ``SVD bound''. For any workload $\W$, the SVD bound is determined by $\W^T\W$ and can be computed directly from it (which can be more efficient):  

\begin{proposition}\label{prop:svdbequiv}
Given $n\times n$ matrix $\W^T\W$. 
\begin{equation*} 
\svdb(\W)=\frac{1}{n}(\sum_{i=1}^{n}\sqrt d_i)^2.
\end{equation*}
where $d_1, \ldots, d_n$ are the eigenvalues of $\W^T\W$.
\end{proposition}

The SVD bound satisfies equivalence properties analogous to (i), (ii), (iii), and (iv) in Prop.~\ref{prop:equiverr} and inequality (i) in Prop.~\ref{prop:errorinequal}. However, it does not satisfy properties related to column projection, as shown in the following counter-example.  

\begin{example}\label{exp:looseness}
Consider a $2\times n$ workload $\W$ consisting of queries $[1,0,\ldots, 0]$ and $[t,t,\ldots, t]$. Let $\mu$ be the column projection w.r.t. the first cell condition of $\W$. When $n>8$ and $t<1/8$, $\svdb(\W) < \svdb(\mu(\W))$.
\end{example}

According to Prop~\ref{prop:errorinequal}, column projections reduce the minimum error. Therefore, the SVD bound on any column projection of $\W$ also constitutes a lower bound for the minimum error of $\W$.  Because of this we extend the simple SVD bound in the following way.  Recall that $\Mu_n$ is the set of all column projections.

\begin{definition}\label{def:esvdb}
Given an $m\times n$ workload $\W$ and $U\subseteq\Mu_n$. The singular value bound of $\W$ w.r.t. $U$, denoted by $\svdb_{U} (\W)$ is defined as
\[\svdb_{U}(\W)=\max_{\mu\in U}\svdb(\mu(\W)).\]
In particular, if $U=\Mu_n$, we call this bound the {\em supreme singular value bound}, denoted $\ssvdb(\W)$.
\end{definition}

According to Prop.~\ref{prop:errorinequal} and Thm.~\ref{thm:singularvaluebound}, for any $U\subseteq\Mu_n$, $\svdb_U(\W)$ provides a lower bound on $\minerror(\W)$.
\begin{corollary}
Given an $m\times n$ workload $\W$, and for any $U\subseteq\Mu_n$
\begin{align*}
\minerror(\W)&\geq \max_{\mu \in U}\minerror(\mu(\W))\\
&\geq P(\epsilon, \delta)\svdb_U(\W),
\end{align*}
where $P(\epsilon, \delta)=\frac{2\log(2/\delta)}{\epsilon^2}$.
\end{corollary}

The supreme SVD bound satisfies all of the error equivalence and containment properties, analogous to those of Prop.~\ref{prop:equiverr} and Prop.~\ref{prop:errorinequal}, as stated below, both of which are proved in App. \ref{app:proof:svdb}.

\begin{restatable}{theorem}{thmsvdbequiv}\label{thm:svdbequiv}
Given an $m_1 \times n_1$ workload $\W_1$ and an $m_2 \times n_2$ workload $\W_2$, the following conditions imply that $\ssvdb(\W_1)=\ssvdb(\W_2)$:
\begin{enumerate} \itemsep 0in
\item[(i)] $\W_1 \equiv \W_2$
\item[(ii)] $\W_1=\Q\W_2$ for some orthogonal matrix $\Q$.
\item[(iii)] $\W_2$ results from permuting the rows of $\W_1$.
\item[(iv)] $\W_2$ results from permuting the columns of $\W_1$.
\item[(v)] $\W_2$ results from column projection of $\W_1$ on all of its nonzero columns.
\end{enumerate}
\end{restatable}

\begin{restatable}{theorem}{thmsvdbinequal}\label{thm:svdbinequal}
Given an $m_1 \times n_1$ workload $\W_1$ and an $m_2 \times n_2$ workload $\W_2$, the following conditions imply that $\ssvdb(\W_1) \leq\ssvdb(\W_2)$:
\begin{enumerate} \itemsep 0in
\item[(i)] $\W_1 \subseteq \W_2$
\item[(ii)] $\W_1$ is a column projection of $\W_2$.
\end{enumerate}
\end{restatable}

 
While Theorems \ref{thm:svdbequiv} and \ref{thm:svdbinequal} show that $\ssvdb(\W)$ matches all the properties of $\minerror(\W)$, we often wish to avoid considering all possible column projections as required in the computation of $\ssvdb(\W)$.  In many cases,  using $\svdb(\W)$ as our lower bound provides good results.  In other cases, we can choose an appropriate set of column projections to get a good approximation to the supreme SVD bound.  We provide empirical evidence for this in the following example, along with an application of our bound to range and predicate workloads which have been studied in prior work.  The bound allows us to evaluate, for the first time, how well existing solutions approximate the minimum achievable error under $(\epsilon,\delta)$-differential privacy.

\begin{example} In Table~\ref{tab:errcmp} we consider three workloads, each consisting of all multi-dimensional range queries for a different dimension set, along with a workload of all predicate queries.  We report $\svdb(\W)$ and its ratio with $\svdb_U(\W)$ where $U$ contains projections onto all possible ranges over the domain, showing that they are virtually indistinguishable.   
\begin{table*}[t]
\centering
\small
\begin{tabular}{c|c|c|cccc}
\multirow{2}{*}{Example Workload, $\W$} & \multirow{2}{*}{$\svdb(\W)$} & \multirow{2}{*}{$\svdb_U(\W)$} & \multicolumn{4}{c}{Error, as ratio to $P(\epsilon, \delta)\svdb(\W)$}\\ 
\cline{4-7} & &  & Identity & Hierarchical & Wavelet & Eigen Design\\
\hline $\allrange(2048)$ & $3.034\times 10^7$ & $1.001$ &  $47.25$ & $1.776$ &$1.545$ & $1.028$\\
$\allrange(64, 32)$ & $2.261\times 10^7$ & $1.000$ & $12.11$ & $2.996$ & $1.899$ & $1.107$\\
$\allrange(2,2,2,2,2,2,2,2,2,2)$ & $5.242\times 10^5$ & $1.000$ & $2.000$ & $2.000$ & $2.000$ & $1.000$ \\
$\allpred(1024)$ & $4.885\times 10^{156}$ & $1.000$ & $1.884$  & $3.464$ & $6.292$  & $1.000$
\end{tabular}
\vspace*{-1ex}
\caption{Four example workloads, their singular value bounds, and their error rates under common strategies and strategies proposed in prior work.}
\label{tab:errcmp}
\vspace*{-2ex}
\end{table*}

We also compute the actual error introduced by several well-known strategies: the identity strategy, the hierarchical strategy~\cite{Hay:2010Boosting-the-Accuracy}, and the wavelet strategy~\cite{xiao2010differential}, as well as a strategy generated by the Eigen-design mechanism~\cite{chaopvldb12}.  These results reveal the quality of these approaches by their ratio to $\svdb(\W)$. For example, from the table we can conclude that the Eigen-design mechanism and wavelet strategies have error at most $1.5$ to $3$ times the optimal for range workloads, but perform worse on the predicate queries.  The identity strategy is far from optimal on low dimensional range queries, but better on high dimensional range queries and predicate queries.  
\end{example}

\subsection{Proof of the SVD bound}\label{sec:svdb:proof}

We now describe the proof of Theorem~\ref{thm:singularvaluebound}.  The key to the proof is an important property of the optimal strategy for the $(\epsilon,\delta)$ matrix mechanism.  As shown in Lemma \ref{lem:columnuniform}, among the  optimal strategies for a workload $\W$, there is always a strategy $\A$ that has the same sensitivity for every cell condition (i.e. in every column). We use $\mathcal{A}_\W$ to denote the set that contains all strategies that satisfy $\W\A^\plus\A=\W$ and have the same sensitivity for every cell condition.

Recall that the sensitivity of strategy $\A$ (Def. \ref{def:l2sens}) is the maximum $L_2$ column norm of $\A$.  The square of the sensitivity is also equal to the maximum diagonal entry of $\A^T\A$.  By using Lemma \ref{lem:columnuniform}, the sensitivity of $\A$ can instead be computed in terms of the trace of the matrix $\A^T\A$ and minimizing the error of $\W$ with this alternative expression of the sensitivity leads to the SVD bounds.  Ultimately, to achieve the SVD bounds, a strategy $\A$ must simultaneously (i) minimize the error of $\W$ with the sensitivity computed in terms of the $\tr(\A^T\A)$, and (ii) have $\A\in\mathcal{A}_\W$. Such a strategy may not exist for every possible $\W$ and therefore the SVD bounds only serve as lower bounds to the minimal error of $\W$.

\begin{restatable}{mylemma}{lemcolumnuniform}\label{lem:columnuniform}
Given a workload $\W$, there exists a strategy $\A\in\mathcal{A}_\W$ such that $\error{\A}{\W}=\minerror(\W)$.
\end{restatable}

%
\begin{restatable}{mylemma}{lemdiagmin}\label{lem:diagmin}
Let $\D$ be a diagonal matrix with non-negative diagonal entries and $\P$ be an orthogonal matrix whose column equals to $\p_1,\p_2,\ldots,\p_n$. 
\[\tr(\D)\leq\sum_{i=1}^n||\D\p_i||_2.\]
\end{restatable}

The proofs of both lemmas are in App. \ref{app:proof:svdb}. 

Theorem~\ref{thm:singularvaluebound} can hence be proved using the lemmas above.
\begin{proof}
For a given workload $\W$, according to Lemma~\ref{lem:columnuniform}, it has an optimal strategy matrix $\A\in\mathcal{A}_\W$, whose sensitivity can then be computed as $\sens{\A}^2=\frac{1}{n}||\A||_F^2$.

Let $\W=\Q_\W\LambdaB_\W\P_\W$ and $\A=\Q_\A\LambdaB_\A\P_\A$ be the singular decomposition of $\W$ and $\A$, respectively. We have:
\begin{eqnarray}
&&\min_{\A:\,\W\A^\plus\A=\W}\sens{\A}^2||\W\A^\plus||_F^2\nonumber\\
&=&\min_{\A\in\mathcal{A}_\W}\frac{1}{n}||\A||_F^2||\W\A^\plus||_F^2\nonumber\\
&=&\frac{1}{n}\min_{(\LambdaB_\A\P_\A)\in\mathcal{A}_\W}||\LambdaB_\A||_F^2||\LambdaB_\W\P_\W\P_\A^T\LambdaB_\A^\plus||_F^2\nonumber\\
&\geq&\frac{1}{n}\min_{\begin{smallmatrix}\LambdaB_\A, \P_\A\\
\LambdaB_\W\LambdaB_\A^\plus\LambdaB_\A^=\LambdaB_\W\end{smallmatrix}}
||\LambdaB_\A||_F^2||\LambdaB_\W\P_\W\P_\A^T\LambdaB_\A^\plus||_F^2\label{eqn:totracecond}\\
&\geq&\frac{1}{n}\min_{\P_\A}(\sum_{i=1}^n ||\LambdaB_\W\p_i||_2)^2\label{eqn:cauchy}\\
&\geq&\frac{1}{n}(\sum_{i=1}^n\lambdaB_i)^2,\label{eqn:minap}
\end{eqnarray}
where $\p_i$ is the $i$-th column of matrix $\P_\W\P_\A^T$,  the inequality in (\ref{eqn:cauchy}) is based on the Cauchy-Schwarz inequality and the inequality in (\ref{eqn:minap}) comes from Lemma~\ref{lem:diagmin}.

The equal sign in (\ref{eqn:cauchy}) is satisfied if and only if $\LambdaB_\A\propto\sqrt{\LambdaB_\W}$. Therefore to achieve equality in (\ref{eqn:cauchy}) and (\ref{eqn:minap}) simultaneously, we need $\A\propto\Q\sqrt{\LambdaB_\W}\P_\W$ for any orthogonal matrix $\Q$. Moreover, (\ref{eqn:totracecond}) is true if and only if $\A\in\mathcal{A}_\W$, which may not be satisfied when $\A\propto\Q\sqrt{\LambdaB_\W}\P_\W$, therefore the SVD bound only gives an lower bound to the minimum total error.
\end{proof}

Intuitively, the SVD bound is based on the assumption that the error can be evenly distributed to all the cells, which may not be achievable in all the cases.  The supreme SVD bound considers only the case that the error can be evenly distributed to some of the cells and therefore may be tighter than the SVD bound.

\subsection{Bounding $\minerror(\W)$ Under the $\epsilon$-Matrix Mechanism}\label{sec:svdb:l1}

The SVD bound is defined for the $(\epsilon, \delta)$-matrix mechanism, so it is natural to consider extending these results to the $\epsilon$-matrix mechanism. Prop.~\ref{prop:totalerror}  can be adopted to the $\epsilon$-matrix mechanism, which uses Laplace noise, an alternative privacy parameter, $P(\epsilon)=1/\epsilon^2$, and measures the sensitivity of $\A$ as the largest $L_1$ norm of the columns of $\A$. For any vector, its $L_1$ norm is always greater than or equal to its $L_2$ norm. Given a workload $\W$ and a strategy matrix $\A$, $P(\epsilon)\sens{\A}^2||\W\A^\plus||_F^2$ provides a lower bound to $\error{\A}{\W}$ under the $\epsilon$-matrix mechanism. Therefore, error under the $\epsilon$-matrix mechanism is also bounded below by $\svdb(\W)$. 

When the number of queries in a workload is no more than the domain size, Bhaskara et al.~\cite{bdkt12stoc} presented the following lower bound of error for any data-independent $\epsilon$-differential privacy mechanism.
\begin{theorem}[\cite{bdkt12stoc}]
Given an $m\times n$ workload $\W$ with $m\leq n$, let convex body $\K=\W\B_1^n$, where $\B_1^m$ is the $m$-dimensional $L_1$ ball. Let $\P_1,\ldots,\P_t$ be projection operators to a collection of $t$ mutually orthogonal subspaces of $\mathbb{R}^m$ of dimension $m_1,\ldots, m_t$ respectively. Then the error of answering $\W$ under any data-independent $\epsilon$-differentially private mechanism must be at least
$$\Omega\left(\sum_i\frac{m_i^3}{\epsilon^2}\mbox{Vol}_{m_i}(\P_i\K)^{2/m_i} \right),$$
where $\mbox{Vol}_{m_i}(\P_i\K)$ is the volume of the convex body $\P_i\K$ in $m_i$ dimensional space.
\end{theorem}

In particular, when $\P_i$ are the projections to the singular vectors of $\W$,  we can formulate the bound above using singular values of $\W$.
\begin{corollary}\label{cor:geobound}
Given an $m\times n$ workload $\W$ with $m\leq n$, the error of answering $\W$ under any data-independent $\epsilon$-differentially private mechanism must be at least
$$\Omega\left(\sum_i^n \frac{\lambda_i^2}{\epsilon^2}\right), $$
where $\lambda_1, \ldots, \lambda_n$ are singular values of $\W$.
\end{corollary}

When $m\leq n$, we can compare the lower bound in Corollary~\ref{cor:geobound} with the SVD bound under the $\epsilon$-matrix mechanism. It is clear that the bound in Corollary~\ref{cor:geobound} is tighter unless all singular values of $\W$ are equal. When $m>n$, the quality of the SVD bound under the $\epsilon$-matrix mechanism is not yet known. The discussion on the tightness and looseness of the SVD bound in the next section is based on the $(\epsilon, \delta)$-matrix mechanism and cannot be extended to the $\epsilon$-matrix mechanism directly.

\section{Analysis of the SVD Bound}\label{sec:svdb:theory}
In this section, we analyze the accuracy of the SVD bound as an approximation of the minimum error for a workload. We study the sufficient and necessary conditions under which the SVD bound is tight. In addition, we show the minimum error is equal to the bound over a specific class of workloads called variable-agnostic workloads and then generalize the result to the widely-studied class of data cube workloads. For both classes, strategies that achieve the minimum error can be constructed, as a by-product of the proof of the SVD bound.

We then show that the bound may be loose, underestimating the minimal error for some workloads.  The worst case of looseness of the SVD bound is presented in Section~\ref{sec:svdb:looseness}, along with a formal estimate of the quality of the bound.  We conclude this section with an example demonstrating empirically that error rates close to the lower bound can be achieved for workloads consisting of multi-dimensional range queries. 

The proofs to the theorems in this section can be found in App. \ref{app:proof:theory}.


\subsection{The Tightness of the SVD Bound}\label{sec:svdb:tightness}

The circumstances under which the SVD bound is tight arise directly from inspection of the proof presented in Sec. \ref{sec:svdb:proof}.  In particular, we noted the conditions that make the inequalities in equations (\ref{eqn:totracecond}), (\ref{eqn:cauchy}) and (\ref{eqn:minap}) actually equal.  Those conditions are equivalent to a straightforward property of $\W^T\W$:

\begin{restatable}{theorem}{thmsvdtightness}\label{thm:svdtightness}
Given workload $\W$, $\svdb(\W)$ is tight if and only if the diagonal entries of $\sqrt{\W^T\W}$ are all equal.
\end{restatable}

There are workloads that satisfy the condition in Thm~\ref{thm:svdtightness}. Here we present one such special class of workloads, called {\em variable-agnostic workloads}, in which the queries on each cell are fully symmetric and swapping any two cells does not change $\W^T\W$.

\begin{definition}[Variable-agnostic workload]~\\
A workload $\W$ is {\em variable-agnostic} if $\W^T\W$ is unchanged when we swap any two columns of $\W$.
\end{definition}

For any variable-agnostic workload $\W$, $\W^T\W$ has the following special form: for some constants $a$ and $b$ such that $a>b$, all diagonal entries of $\W^T\W$ are equal to $a$ and the remaining entries of $\W^T\W$ are equal to $b$.
%

The following theorem shows that any variable-agnostic workload $\W$ satisfies the condition in Thm~\ref{thm:svdtightness}. Furthermore, we also demonstrate the closed form expression of the SVD bound in case that $n$ is a power of $2$.


\begin{restatable}{theorem}{thmvaragn} \label{thm:varagn}
The SVD bound is tight for any variable-agnostic workload $\W$. In addition, when $n=2^k$ for any nonnegative integer $k$, $\svdb(\W)=\frac{1}{n}(\sqrt{a+(n-1)b}+(n-1)\sqrt{a-b})^2$, where $a$ is the value of diagonal entries of $\W^T\W$ and $b$ is the value of off-diagonal entries of $\W^T\W$.
\end{restatable}

As a concrete example, the workload $\allpred(n)$ is variable-agnostic, and therefore we can construct its optimal strategy and compute the error rate directly.

\begin{corollary}\label{col:svdballpred}
The SVD bound is tight for the workload $\allpred(n)$. In addition, when $n=2^k$ for any nonnegative integer $k$,  $\svdb(\allpred(n))=\frac{2^{n-2}}{n}(n-1+\sqrt{n+1})^2$.
\end{corollary}

For variable-agnostic workloads, using a naive strategy like the identity matrix or the workload itself results in total error equal to $na$ and the ratio by which the error is reduced using the strategy in Thm.~\ref{thm:varagn} is approximately $1-\frac{b}{a}$.  In the case of $\allpred(n)$, the ratio is at least as low as $0.5$, which occurs when $n$ is very large.


Another family of workloads for which the SVD bound is tight are those consisting of sets of data cube queries. A data cube workload consists of one or more cuboids, each of which contains all aggregation queries on all possible values of the cross-product of a set of attributes. Here we also consider the case that each cuboid can have its own weight, so that higher weighted queries will be estimated more accurately than lower weighted ones.

\begin{restatable}{theorem}{thmdatacube} \label{thm:datacube}
The SVD bound is tight for any weighted data cube workload $\W$. 
\end{restatable}

Data cube workloads (a special case of marginal workloads) have been studied by the differential privacy community in both theory and practice~\cite{barak2007privacy, Ding:2011fk, kasiviswanathan2010price}. Barak et al.~\cite{barak2007privacy} use the Fourier basis as a strategy for workloads consisting of marginals while Ding et al.~\cite{Ding:2011fk} proposed an approximation algorithm for data cube workloads. Thm.~\ref{thm:datacube} shows that under $(\epsilon, \delta)$-differential privacy we can now directly compute the optimal strategy, obviating the need to use an approximation algorithm or blindly relying on the Fourier basis for workloads of this type. The result in \cite{kasiviswanathan2010price}, however, involves data-dependent techniques and the comparison between  \cite{kasiviswanathan2010price} to the SVD bound relies on a thorough analysis of the spectral properties of data cube workloads, which is a direction of future work.

\subsection{The Looseness of  the SVD Bound}\label{sec:svdb:looseness}
%
The SVD bound can also underestimate the minimum error when the workload is highly skewed.  For example, the SVD bound does not work well when the sensitivity of one column in the workload is overwhelmingly larger than others.  Recall the workload in Example~\ref{exp:looseness}, when $t\rightarrow0$, the SVD bound will underestimate the total error by a factor of $n$. This is caused by the underestimate of the sensitivity of $\A$ considered in equation~(\ref{eqn:totracecond}) in the proof of Thm.~\ref{thm:singularvaluebound}.  

Since the proof of Thm.~\ref{thm:singularvaluebound} constructs a concrete strategy, one way to measure the looseness of the SVD bound is to estimate its ratio to the actual error introduced by this strategy. Note that the sensitivity of the strategy is the only part of the SVD bound that is underestimated.  The square of the sensitivity is the maximum diagonal entry of matrix $\A^T\A$, rather than the estimate given by $\tr(\A^T\A)/n$.  The ratio between the actual sensitivity and the estimated sensitivity bounds the looseness of the SVD bound, as shown by the following theorem.
\begin{restatable}{theorem}{thmsvdratio}\label{thm:svdratio}
Given an $m\times n$ workload $\W$. Let $d_0$ be the maximum diagonal entry of $\sqrt{\W^T\W}$.
\[\minerror(\W)\leq \frac{nd_0P(\epsilon, \delta)\svdb(\W)}{\tr(\sqrt{\W^T\W})},\]
where $P(\epsilon, \delta)=\frac{2\log(2/\delta)}{\epsilon^2}$.
\end{restatable}

According to Thm.~\ref{thm:svdratio}, the approximate ratio of the SVD bound corresponds to the ratio between $d_0$,  the largest diagonal entry of $\sqrt{\W^T\W}$ and the trace of $\sqrt{\W^T\W}$, which is equal to the sum of all singular values of $\W$. This ratio, although it is upper-bounded by the ratio between the largest singular value of $\W$ and the sum of all singular values of $\W$, is much closer to 1 than the ratio between singular values. As a consequence, the skewness in singular values does not always lead to a bad approximation ratio for the SVD bound. For example, for variable-agnostic workloads, the largest singular value can be arbitrarily larger than the rest of the singular values, while the SVD bound is tight. Instead, the cases where the SVD bound has high approximation ratio, such as the one in Example~\ref{exp:looseness}, are due to the skewness of singular value of $\W$ and the particular distribution of singular vectors. The supreme SVD bound can help us to avoid some of these worst cases, but there is no guarantee of the quality of the bound with more sophisticated cases.

Nevertheless, for many common workloads, empirical evidence suggests that the SVD bound is quite close to the minimal error.  The following example provides a comparison between the SVD bound and achievable error for a few common workloads.

\begin{example}
Returning to Table~\ref{tab:errcmp}, we observe empirical evidence that for range and predicate workloads, there are strategies that come quite close to the SVD bound.  The last column of Table \ref{tab:errcmp} lists the error for the Eigen-design mechanism \cite{chaopvldb12}, which attempts to find approximately optimal strategies for any given workload by computing optimal weights for the eigenvectors of the workload.  This algorithm is able to find a strategy whose error is within a factor of $1.028$ and $1.107$ of optimal for $\allrange(2048)$ and $\allrange(64,32)$, respectively. 
\end{example}


\section{Comparison of Mechanisms}\label{sec:comparison}

The matrix mechanism is a data-independent mechanism: the noise distribution (and therefore error) depends only on the workload and not on the particular input data.  This makes it possible to process the workload once and apply the mechanism efficiently to any dataset.  On the other hand, data-independent mechanisms lack the flexibility to exploit specific properties of individual datasets.  In this section, we use the SVD bound to compare the error bounds of the matrix mechanism with error bounds of other mechanisms that are data-dependent. 


\subsection{Asymptotic Estimation of the SVD bound}
Before the comparison, we first convert the SVD bound into an error measure that can be directly related to other bounds in the literature. We assume all queries in the workload have sensitivity at most one and estimate the SVD bound as a function of the domain size $n$ and the number of queries $m$.  Recall that the error in previous sections is defined as the total mean squared error of the queries. We introduce a new measure of error which bounds the maximum absolute error of the workload queries by $\alpha$ with high probability (controlled by $\beta$).
\begin{definition}[$(\alpha, \beta)$-Accurate~\cite{Gupta:2012uq}]
\hspace*{-1ex}Given a workload $\W$, an algorithm $\alg$ is $(\alpha, \beta)$-accurate if, for any uniformly drawn data vector $\x$, with a probability of at least $1-\beta$, $\max_{\q\in\W}|\alg(\q, \x)-\q\x|\leq\alpha$.
\end{definition}

Since the SVD bound measures total error (rather than max error), here we modify the $(\alpha, \beta)$-accuracy by bounding the root mean squared error of the workload.
\begin{definition}[RMS-$(\alpha, \beta)$-Accurate]
Given a \\workload $\W$, an algorithm $\alg$ is RMS-$(\alpha, \beta)$-accurate if, for any uniformly drawn data vector $\x$, with a probability of at least $1-\beta$, $\sqrt{\sum_{\q\in\W}||\alg(\q, \x)-\q\x||^2/|\W|}\leq\alpha$.
\end{definition}

\begin{restatable}{theorem}{thmasysvd}
Given an $m\times n$ workload $\W$, if the \\$\svdb(\W)$ is asymptotically tight, then there exists a strategy under which the matrix mechanism is RMS-$(\alpha, \beta)$-accurate, where
\[\alpha=O\left(\frac{\sqrt{\min(m, n)}\sqrt{\log(2/\delta)\log(\sqrt{\pi/2}/\beta)}}{\epsilon}\right).\]
\end{restatable}

The proof of the theorem can be found in App. \ref{app:proof:comparison}.

Recall the discussion in Sec.~\ref{sec:svdb:tightness} indicates that the SVD bound is tight or almost tight for many common workloads. Thus, it is reasonable to compare the asymptotic estimate of the SVD bound to the error introduced by other mechanisms.

\subsection{Comparison of Error Bounds}
Here we compare our SVD bound with other error bounds from data-dependent mechanisms. We include four competitors each representing fundamentally different mechanisms. The median mechanism~\cite{DBLP:conf/stoc/RothR10} discards candidate data vectors that are inconsistent with historical query answers. The multiplicative weights mechanism (MW)~\cite{hardt2010multiplicative} and the iterative database construction method (IDC)~\cite{Gupta:2012uq} repeatedly update an estimated data vector according to query answers. The boosting method~\cite{DBLP:conf/focs/DworkRV10} maintains a distribution of queries according to the quality of their answers and repeatedly samples queries from the distribution so as to improve their answers. The $(\alpha, \beta)$-accuracy under $(\epsilon, \delta)$-differential privacy for the median and the multiplicative weight mechanism follows the result in \cite{Gupta:2012uq}.

\begin{centering}
\begin{table}[tc]
\centering
\small
\begin{tabular}{|l|l|c|}
\hline
\multicolumn{2}{|c|}{\bf{Mechanism}} & $\mathbf{\alpha}$ \\
\hline 1 & Median~\cite{DBLP:conf/stoc/RothR10} & $O\left(\frac{\sqrt{N}(\log n\log m)^{1/4}\sqrt{t(\log m+t)}}{\sqrt{\epsilon}}\right)$\\
\hline 2 & MW~\cite{hardt2010multiplicative} & $O\left(\frac{\sqrt{N}(\log n)^{1/4}\sqrt{t(\log m+t)}}{\sqrt{\epsilon}}\right)$\\
\hline 3 & IDC~\cite{Gupta:2012uq} & $O\left(\frac{(nN)^{1/4}\sqrt{t(\log m+t)}}{\sqrt{\epsilon}}\right)$\\
\hline 4 & Boosting~\cite{DBLP:conf/focs/DworkRV10} & $\tilde O\left(\frac{\sqrt{N\log n}\cdot t^{3/2}\log^{3/2}m }{\epsilon}\right)$\\ \hline
\hline 5 & SVDB & $O\left(\frac{\sqrt{\min(m, n)}\cdot t}{\epsilon}\right)$\\
\hline
\end{tabular}
\vspace*{-1.5ex}
\caption{For $t\geq 2$, bounds on the $\alpha$ required to achieve $(\epsilon, \exp(-t))$-differential privacy and accuracy measures of: $(\alpha, \exp(-t))$-accuracy (mechanisms 1-4); RMS-$(\alpha, \exp(-t))$-accuracy (mechanism 5).}
\label{tbl:errcmp}
\vspace*{-3ex}
\end{table}
\end{centering}

Table~\ref{tbl:errcmp} summarizes error bounds of different data dependent approaches. In particular, the comparison is over $(\epsilon, \exp(-t))$-differential privacy and $(\alpha, \exp(-t))$-accuracy. \\The workload $\W$ we considered contains $m$ queries with sensitivity no larger than $1$. The database  is of size $N$, which means the sum of all $x_i$'s in the data vector is $N$.

Observing the values of $\alpha$ in Table~\ref{tbl:errcmp}, the matrix mechanism has a greater dependence on $\epsilon$ compared with the median, the multiplicative weights and the iterative database construction methods. In addition, since the matrix mechanism is data-independent, it cannot take advantage of the input dataset so that it always assumes $n=N$. However, when $N$ is sufficiently large ($\Theta(n)$) and $m=O(n)$, the SVD bound is smaller than the error of the Boosting method and can outperform other competitors when $m=\Omega(\exp(t/\epsilon))$.

\subsection{Data-dependency \& the matrix mechanism}
Although the techniques of the matrix mechanism are data-independent, they can be deployed in a data-dependent way, blurring the distinction between mechanism types.  The differentially private domain compression technique~\cite{Li:2011:CMU:2046556.2046581} may be applied to reduce the domain size $n$ to $\Theta(N)$ with an additional $O(\log n)$ noise, which suggests a method for improving the error dependency of the matrix mechanism on $n$.

Further, the optimal strategy matrix used in the matrix mechanism represents the fundamental building blocks of the workload and the matrix mechanism reduces error by using the strategy queries as differentially private observations, instead of the workload queries.  Recent data-dependent approaches can benefit from the same approach.  In fact, \cite{hlm2012nips} selects Fourier basis vectors adaptively in a data dependent manner, but could benefit from selecting from a more efficient strategy matrix.  Therefore, the SVDB bound can serve as a baseline accuracy measure, which may be improved by data-dependent query selection.

\section{Related Work}

The original description of the matrix mechanism \cite{Li:2010Optimizing-Linear} focuses primarily on $\epsilon$-diff\-er\-ent\-ial privacy, with a brief consideration of $(\epsilon,\delta)$-differential privacy.  A number of proposed mechanisms can be formulated as instances of the matrix mechanism:  techniques for accurately answering range queries are presented in \cite{xiao2010differential, Hay:2010Boosting-the-Accuracy}; 
low-order marginals are studied in \cite{barak2007privacy} using a Fourier transformation as the strategy (combined with other techniques for achieving integral consistency) as well as in \cite{Yaroslavtsev13Accurate} by optimally scaling a manually-chosen set of strategy queries; an algorithm for generating good strategies for answering sets of data cube queries is introduced in \cite{Ding:2011fk}; and an algorithm for computing optimal low-rank strategy matrices is presented in \cite{Yuan12Low-Rank}. The lower bound presented in this work provides a theoretical method to evaluate the quality of each of the approaches above, assuming $(\epsilon, \delta)$-differential privacy.

In recent work, Nikolov et al.~\cite{Nikolov:arXiv1212.0297} propose an algorithm whose error is within a ratio of $O(\log^2\rank(\W)\log(1/\delta))$ to the optimal error under {\em any} data-independent $(\epsilon, \delta)$-diff-erentially private mechanism (not limited to instances of the matrix mechanism).  Their algorithm is in fact a special case of the  $(\epsilon, \delta)$-matrix mechanism, so this approximation ratio also bounds the ratio between the SVD bound and the minimum achievable error of {\em all} possible data-independent $(\epsilon,\delta)$-differential private mechanisms.


Blum et al.~\cite{blum2008a-learning} describe a very general mechanism for synthetic data release, in which error rates are related to the VC dimension of the workload. However, for many workloads of linear queries, VC dimension is too coarse-grained to provide a useful measure of workload error complexity. For example, the VC dimension for any workload of $d$-dimensional range queries that can not be embedded into $(d-1)$-dimensional spaces is always $d+1$, despite the fact that such workloads could have very different achievable error rates.

Hardt et al.~\cite{Hardt:2010On-the-Geometry-of-Differential} present a lower bound on error for low rank workloads. Similar to the SVD bound, this geometric bound can also be represented as a function of the singular values of the workload. In particular, the bound uses the geometric average of the singular values rather than the algebraic average in the SVD bound. The geometric bound provides a more general guarantee since it is a lower bound on all $\epsilon$-differential privacy mechanisms.  But it is not directly comparable with the SVD bound since it bounds the mean absolute error rather than mean squared error in the SVD bound.  
Lower and upper bounds on answering all $k$-way marginals with a data dependent mechanism are discussed in \cite{kasiviswanathan2010price}. Though it is clear that the SVD bound is tight in the case of all $k$-way marginals (since it is a special case of data cube) comparison with \cite{kasiviswanathan2010price} requires a careful analysis of the singular values of workloads of $k$-way marginals and is a direction for future investigation.


There are also error bounds from data-dependent mechanisms, some of which we have compared with in Sec. \ref{sec:comparison}.  A data-dependent approach for range queries is described in \cite{Cormode:2012Spatial}. The median mechanism~\cite{DBLP:conf/stoc/RothR10} drops data vectors that are inconsistent with query answers in each step.  Dwork et al.~\cite{DBLP:conf/focs/DworkRV10} samples linear queries in each step and modifies the sample distribution with the new query answers. In \cite{hardt2010multiplicative, hlm2012nips, Gupta:2012uq}, the authors recursively update the estimated data vector to reduce the error on linear queries. More generally, Dwork et al. provide an error bound using an arbitrary differentially private mechanism~\cite{dwork2009complexity} but not specifically for linear counting queries. Those analyses lead to smaller error than the matrix mechanism over sparse databases by analyzing the properties of the underlying database. Thus their error bounds reflect the connection between workloads and databases  but cannot lead to bounds on error that can be used to characterize the error complexity of workloads.



\vspace*{-1ex}
\section{Conclusion}
We have shown that, for a general class of $(\epsilon,\delta)$-differentially private algorithms, the error rate achievable for a set of queries is determined by the spectral properties of the queries when they are represented in matrix form.  The result is a lower bound on error which is a simple function of the eigenvalues of the query matrix.  The bound can be used to assess the quality of a number of existing differentially private algorithms, to directly construct error-optimal strategies in some cases, to compare the hardness of query sets, and to guide users in the design of query workloads.  

\subsection*{Acknowledgements}
We appreciate the helpful comments of the anonymous reviewers. Li was supported by NSF CNS-1012748. Miklau was partially supported by NSF CNS-1012748, NSF CNS-0964094, and the European Research Council under the Webdam grant.
\vspace*{-1ex}

\bibliographystyle{abbrv} 
{ 
\bibliography{paper} 

\begin{thebibliography}{10}

\bibitem{barak2007privacy}
B.~Barak, K.~Chaudhuri, C.~Dwork, S.~Kale, F.~McSherry, and K.~Talwar.
\newblock Privacy, accuracy, and consistency too: A holistic solution to
  contingency table release.
\newblock In {\em PODS}, 2007.

\bibitem{ben2003generalized}
A.~Ben-Israel and T.~Greville.
\newblock {\em Generalized inverses: Theory and applications}, volume~15.
\newblock Springer, 2003.

\bibitem{bdkt12stoc}
A.~Bhaskara, D.~Dadush, R.~Krishnaswamy, and K.~Talwar.
\newblock Unconditional differentially private mechanisms for linear queries.
\newblock In {\em STOC}, pages 1269--1284, New York, NY, USA, 2012.

\bibitem{blum2008a-learning}
A.~Blum, K.~Ligett, and A.~Roth.
\newblock A learning theory approach to non-interactive database privacy.
\newblock In {\em STOC}, pages 609--618, 2008.

\bibitem{Cormode:2012Spatial}
G.~Cormode, M.~Procopiuc, E.~Shen, D.~Srivastava, and T.~Yu.
\newblock Differentially private spatial decompositions.
\newblock {\em ICDE}, pages 20--31, 2012.

\bibitem{Ding:2011fk}
B.~Ding, M.~Winslett, J.~Han, and Z.~Li.
\newblock Differentially private data cubes: optimizing noise sources and
  consistency.
\newblock In {\em SIGMOD}, pages 217--228, 2011.

\bibitem{Dinur03Revealing}
I.~Dinur and K.~Nissim.
\newblock Revealing information while preserving privacy.
\newblock In {\em PODS}, pages 202--210, 2003.

\bibitem{Dwork:2006Our-Data-Ourselves:}
C.~Dwork, K.~Kenthapadi, F.~McSherry, I.~Mironov, and M.~Naor.
\newblock Our data, ourselves: Privacy via distributed noise generation.
\newblock In {\em EUROCRYPT}, pages 486--503, 2006.

\bibitem{Dwork:2006Calibrating-Noise}
C.~Dwork, F.~McSherry, K.~Nissim, and A.~Smith.
\newblock Calibrating noise to sensitivity in private data analysis.
\newblock In {\em TCC}, pages 265--284, 2006.

\bibitem{dwork2009complexity}
C.~Dwork, M.~Naor, O.~Reingold, G.~Rothblum, and S.~Vadhan.
\newblock On the complexity of differentially private data release: efficient
  algorithms and hardness results.
\newblock In {\em STOC}, pages 381--390, 2009.

\bibitem{DBLP:conf/focs/DworkRV10}
C.~Dwork, G.~N. Rothblum, and S.~P. Vadhan.
\newblock Boosting and differential privacy.
\newblock In {\em FOCS}, pages 51--60, 2010.

\bibitem{fulton2000eigenvalues}
W.~Fulton.
\newblock Eigenvalues, invariant factors, highest weights, and schubert
  calculus.
\newblock {\em Bulletin of the AMS}, 37(3):209--250, 2000.

\bibitem{Gupta:2012uq}
A.~Gupta, A.~Roth, and J.~Ullman.
\newblock Iterative constructions and private data release.
\newblock In {\em TCC}, pages 339--356, 2012.

\bibitem{hlm2012nips}
M.~Hardt, K.~Ligett, and F.~McSherry.
\newblock A simple and practical algorithm for differentially private data
  release.
\newblock In {\em NIPS}, pages 2348--2356, 2012.

\bibitem{hardt2010multiplicative}
M.~Hardt and G.~Rothblum.
\newblock {A multiplicative weights mechanism for privacy-preserving data
  analysis}.
\newblock In {\em FOCS}, pages 61--70, 2010.

\bibitem{Hardt:2010On-the-Geometry-of-Differential}
M.~Hardt and K.~Talwar.
\newblock On the geometry of differential privacy.
\newblock In {\em STOC}, pages 705--714, 2010.

\bibitem{Hay:2010Boosting-the-Accuracy}
M.~Hay, V.~Rastogi, G.~Miklau, and D.~Suciu.
\newblock Boosting the accuracy of differentially-private histograms through
  consistency.
\newblock {\em PVLDB}, 3(1-2):1021--1032, 2010.

\bibitem{kasiviswanathan2010price}
S.~Kasiviswanathan, M.~Rudelson, A.~Smith, and J.~Ullman.
\newblock The price of privately releasing contingency tables and the spectra
  of random matrices with correlated rows.
\newblock In {\em STOC}, pages 775--784, 2010.

\bibitem{Li:2010Optimizing-Linear}
C.~Li, M.~Hay, V.~Rastogi, G.~Miklau, and A.~McGregor.
\newblock Optimizing linear counting queries under differential privacy.
\newblock In {\em PODS}, pages 123--134, 2010.

\bibitem{chaopvldb12}
C.~Li and G.~Miklau.
\newblock An adaptive mechanism for accurate query answering under differential
  privacy.
\newblock {\em PVLDB}, 5(6):514--525, 2012.

\bibitem{Li:2011:CMU:2046556.2046581}
Y.~D. Li, Z.~Zhang, M.~Winslett, and Y.~Yang.
\newblock Compressive mechanism: utilizing sparse representation in
  differential privacy.
\newblock In {\em WPES}, pages 177--182, 2011.

\bibitem{McSherry:2009fk}
F.~McSherry and I.~Mironov.
\newblock {Differentially Private Recommender Systems : Building Privacy into
  the Netflix Prize Contenders}.
\newblock In {\em SIGKDD}, pages 627--636, 2009.

\bibitem{Nikolov:arXiv1212.0297}
A.~Nikolov, K.~Talwar, and L.~Zhang.
\newblock The geometry of differential privacy: the sparse and approximate
  cases.
\newblock {\em CoRR}, arXiv/1212.0297, 2012.

\bibitem{nissim2007smooth}
K.~Nissim, S.~Raskhodnikova, and A.~Smith.
\newblock Smooth sensitivity and sampling in private data analysis.
\newblock In {\em STOC}, pages 75--84, 2007.

\bibitem{rastogi2007the-boundary}
V.~Rastogi, D.~Suciu, and S.~Hong.
\newblock The boundary between privacy and utility in data publishing.
\newblock In {\em VLDB}, number 531-542, 2007.

\bibitem{DBLP:conf/stoc/RothR10}
A.~Roth and T.~Roughgarden.
\newblock Interactive privacy via the median mechanism.
\newblock In {\em STOC}, pages 765--774, 2010.

\bibitem{xiao2010differential}
X.~Xiao, G.~Wang, and J.~Gehrke.
\newblock Differential privacy via wavelet transforms.
\newblock In {\em ICDE}, pages 225--236, 2010.

\bibitem{xiaodifferentially}
Y.~Xiao, L.~Xiong, and C.~Yuan.
\newblock Differentially private data release through multidimensional
  partitioning.
\newblock In {\em SDM}, pages 150--168, 2010.

\bibitem{Yaroslavtsev13Accurate}
G.~Yaroslavtsev, G.~Cormode, C.~M. Procopiuc, and D.~Srivastava.
\newblock Accurate and efficient private release of datacubes and contingency
  tables.
\newblock In {\em ICDE}, 2013.

\bibitem{Yuan12Low-Rank}
G.~Yuan, Z.~Zhang, M.~Winslett, X.~Xiao, Y.~Yang, and Z.~Hao.
\newblock Low-rank mechanism: Optimizing batch queries under differential
  privacy.
\newblock {\em PVLDB}, 5(11):1136--1147, 2012.

\end{thebibliography}
\normalsize
}

\newpage
\appendix

\vspace*{-.5ex}
\section{An Algebra for Workloads}\label{sec:operation}
\begin{table*}[t]
\small
\hspace*{-5ex}
\centering
\begin{tabular}{c|c|c}

 & $ \svdb$ 
 & $\ssvdb$ 
 \\ \hline

$\W_1 \cup \W_2$ 
& $\sqrt{\svdb(\W_1)} + \sqrt{\svdb(\W_2)} \geq \sqrt{\svdb(\W_1\cup\W_2)}$
& $\sqrt{\ssvdb(\W_1)}+ \sqrt{\ssvdb(\W_2)} \geq \sqrt{\ssvdb(\W_1\cup\W_2)}$ 
\\


$\W_1 \times \W_2$ 
& $\svdb(\W_1)\svdb(\W_2) = \svdb(\W_1\times\W_2)$
& $\ssvdb(\W_1)\ssvdb(\W_2) \leq \ssvdb(\W_1\times\W_2)$ 
\\

\hline
\multicolumn{3}{c}{Predicate Workloads}\\
\hline

$\W_1 \wedge \W_2$ 
& $\svdb(\W_1)\svdb(\W_2) = \svdb(\W_1\wedge\W_2)$
& $\ssvdb(\W_1)\ssvdb(\W_2) \leq \ssvdb(\W_1\wedge\W_2)$ 

 \\ \hline
\end{tabular}
\vspace{-1ex}
\caption{Algebra operators and relations for the simple and supreme singular value bounds.}
\label{tab:wkld2svdb}
\vspace*{-3ex}
\end{table*}

In this section we briefly discuss the relationship between workload operations and the SVD bound. We define basic operators of negation-free relational algebra, union and crossproduct, on workloads and show how our error measure behaves in the presence of these operators.  Many common workloads are the result of combining simpler workloads using these operators.  Thus, the following results can be used to save computation of the SVD bound. In particular, for the crossproduct operation, the computation time for the SVD bound of the crossproduct of two workloads with size $m_1\times n_1$ and $m_2\times n_2$ can be reduced from $O(\min(m_1m_2, n_1n_2)m_1m_2n_1n_2)$ to $O(\min(m_1,n_1)m_1n_1+\min(m_2,n_2)m_2n_2)$.

The proofs of the theorems can be found in App.~\ref{app:proof:operation}. Table~\ref{tab:wkld2svdb}  summarizes the results.


\subsection{Union}\label{sec:operation:union}
The union operation on workloads has the standard meaning for rows of the workload matrix:

\begin{definition}[\sc Union]
Given an $m_1\times n$ workload $\W_1$ and an $m_2\times n$ workload $\W_2$ over the same $n$ cell conditions. $\W_1\cup\W_2$ is the union of $\W_1$ and $\W_2$, the workload consisting of the rows of both $\W_1$ and $\W_2$, without duplicates.
\end{definition}

The relationship between the SVD bounds of workloads and their unions
can be bounded:
\begin{restatable}{theorem}{thmsvdofsum}\label{thm:svbdofsum}
Given an $m_1\times n$ workload $\W_1$ and  an $m_2\times n$ workload $\W_2$ on the same set of $n$ cell conditions.
\begin{align*}\sqrt{\svdb(\W_1)}+\sqrt{\svdb(\W_2)}&\geq\sqrt{\svdb(\W_1\cup \W_2)};\\
\sqrt{\ssvdb(\W_1)}+\sqrt{\ssvdb(\W_2)}&\geq\sqrt{\ssvdb(\W_1\cup \W_2)}.
\end{align*}
\end{restatable}

%
%
%

\subsection{Workload combination}\label{sec:operation:xprod}
Given two workloads over distinct sets of cell conditions, we can combine them to form a workload over the crossproduct of the individual cell conditions. This is most commonly used to combine workloads defined over distinct sets of attributes $\mathbb{B}_1$ and $\mathbb{B}_2$ to get a workload defined over $\mathbb{B}_1 \cup \mathbb{B}_2$.  When we pair individual predicate queries, it is equivalent to pair them conjunctively.  

\begin{definition}[Workload combination] Given an $m_1\times n_1$ workload $\W_1$ defined by cell conditions $\Phi=\cell_1 \dots \cell_{n_1}$ and an $m_2\times n_2$ workload $\W_2$ defined by distinct cell conditions $\Psi = \psi_1 \dots \psi_{n_2}$, a new combined workload $\W$ is defined over cell conditions $\{\cell_i \wedge \psi_j \:|\: \cell_i \in \Phi, \psi_j \in \Psi \}$. For each $\w_1=(w_{1,1}, \ldots, w_{n_1,1})\in \W_1$ and $\w_2=(w_{1,2}, \ldots, w_{n_2,2})\in \W_2$, there is a query $\w \in \W$ accordingly:
\begin{itemize}
\item{\sc(Crossproduct)} If the entry of $\w$ related to each cell condition $\cell_i\wedge\psi_j$ is $w_{1,i}\cdot w_{2,j}$, $\W$ is called the {\em crossproduct} of $\W_1$ and $\W_2$, denoted as $\W_1\times \W_2$.
\item{\sc(Conjunction)} If both $\W_1$ and $\W_2$ consist of predicate queries and the entry of $\w$ related to each cell condition $\cell_i\wedge\psi_j$ is $w_{1,i}\wedge w_{2,j}$, then $\W$ is called the {\em conjunction} of $\W_1$ and $\W_2$, denoted as $\W_1\wedge \W_2$.
\end{itemize} 
\end{definition}
%

The next theorem describes the singular value bound for the crossproduct of workloads:
\vspace*{-1ex}
\begin{restatable}{theorem}{thmsvdofxprod}\label{thm:svdofxprod}
Given an $m_1\times n_1$ workload $\W_1$ and an $m_2\times n_2$ workload $\W_2$ defined on two distinct sets of cell conditions: 
\begin{align*}
\svdb(\W_1 \times \W_2)&=\svdb(\W_1)\svdb(\W_2)\\
\ssvdb(\W_1 \times \W_2)&\geq\ssvdb(\W_1)\ssvdb(\W_2)
\end{align*}
\end{restatable}

The conjunction of predicate queries is a special case of crossproduct, Thus, applying Theorem~\ref{thm:svdofxprod} to predicate workloads we have:
\begin{corollary}\label{col:svdofconj}
Given an $m_1\times n_1$ workload $\W_1$ and an $m_2\times n_2$ workload $\W_2$, both of which consist of predicate queries.
\begin{eqnarray*}
\svdb(\W_1\wedge \W_2)&=&\svdb(\W_1)\svdb(\W_2);\\
\ssvdb(\W_1\wedge \W_2)&\geq&\ssvdb(\W_1)\ssvdb(\W_2).
\end{eqnarray*}
\end{corollary}
\section{Proofs}\label{app:proof}
The following appendices contain proofs that are omitted from the body of the paper. In those proofs, we say $\A$ a strategy on workload $\W$ if $\W\A^\plus\A=\W$.
\subsection{The Extended Matrix Mechanism}\label{app:proof:extendmmech}
In this section we present the proofs related to the extended matrix mechanism. 
\extendmmech*
\begin{proof}
Noticing $\W\A^\plus\A=\W$, 
\begin{align*}\MM_\A(\W,\x) &= \W\x + \W \A^\plus \Nor(\sigma)^m \\
&= \W\A^\plus(\A\x+\Nor(\sigma)^m).
\end{align*}
Since $\sigma=\sens{\A}\sqrt{\log(2/\delta)}/\epsilon$, according to Proposition~\ref{prop:l2diffpriv}, answering $\A\x$ with $\A\x+\Nor(\sigma)^m$ satisfies the $(\epsilon, \delta)$-differential privacy. Thus $\MM_\A(\W, \x)$ satisfies $(\epsilon, \delta)$-differential privacy as well.
\end{proof}

\proptotalerror*
\begin{proof}
According to Def.~\ref{def:wkerror} and Prop.~\ref{def:m-mech}, given a query $\w$ and a strategy $\A$, the mean squared error is
\begin{align*}
\error{\w}{\A}&=\E((\w\x-\hat\w)^2)=\var(\hat\w)=\var(\w \A^\plus \Nor(\sigma)^m)\\
&=\sigma^2||\w\A^\plus||_2^2=P(\epsilon, \delta)\sens{\A}^2||\w\A^\plus||_2^2.
\end{align*}
Therefore for a given workload $\W$, 
\[\error{\W}{\A}=\sum_{\w\in\W}P(\epsilon, \delta)\sens{\A}^2||\w\A^\plus||_2^2=P(\epsilon, \delta)\sens{\A}^2||\W\A^\plus||_F^2.\]
\end{proof}

\subsection{Workload Containment and Equivalence}\label{app:proof:contain}
This section contains the proofs of the relationship between workloads and their minimal errors in Section~\ref{sec:mechanism:equiv}.
\propequiverr*
\begin{proof}
(i): If $\W_1\equiv\W_2$, for any strategy $\A$,
\begin{align*}
||\W_1\A^\plus||_F^2&=\tr((\A^\plus)^T\W_1^T\W_1\A^\plus)\\
&=\tr((\A^\plus)^T\W_2^T\W_2\A^\plus)=||\W_2\A^\plus||_F^2.
\end{align*}
Therefore $\minerror(\W_1)=\minerror(\W_2)$.\\
(ii): It is equivalent with (i).\\
(iii): It is a special case of (ii) where $\Q$ is a permutation matrix.\\
(iv): Let $\P$ is the permutation matrix such that $\W_1\P=\W_2$. For any strategy $\A$ on $\W_1$, $\A\P$ is a strategy of $\W_2$ and $\error{\A}{\W_1}=\error{\A\P}{\W_2}$. \\
(v): Since (iv) is true, we can assume $\W_1=[\W_2, \vect{0}]$. For any strategy matrix $\A_2$ on $\W_2$, $\A_1=[\A_2,\vect{0}]$ is a strategy on $\W_1$ and
\begin{align*}
||\W_1\A_1^\plus||_F^2&=\tr(\A_1^\plus(\A_1^\plus)^T\W_1^T\W_1)\\
&=\tr([\A_2^\plus,\vect{0}]^T[\A_2^\plus,\vect{0}][\W_2, \vect{0}]^T[\W_2, \vect{0}])\\
&=\tr((\A_2^\plus)^T\W_2^T\W_2\A_2^\plus)=||\W_2\A_2^\plus||_F^2.
\end{align*}

For any strategy $\A_1$ on $\W_1$  there is a strategy on $\W_2$ with equal or smaller error formed by deleting corresponding columns from $\A_2$.
\end{proof}

\properrinequal*
\begin{proof}
For (i), let $\W'_2\equiv \W_2$ such that $\W'_2$ contains all rows of $\W_1$. According to Prop.~\ref{prop:equiverr} (iii), we can assume $\W_2'=\left[\begin{smallmatrix}\W_1\\\W_3\end{smallmatrix}\right]$. For any strategy $\A$ on $\W_2$, since $\A$ is also a strategy on $\W_2'$, $\A$ can represent all queries in $\W_1$ as well. Thus $\A$ is a strategy on $\W_1$. In addition, 
\begin{eqnarray*}
&&\error{\A}{\W_2}\\
&=&P(\epsilon,\delta)\sens{\A}^2||\W_2\A^\plus||_F^2\\
&=&P(\epsilon,\delta)\sens{\A}^2||\W_2'\A^\plus||_F^2\\
&=&P(\epsilon,\delta)\sens{\A}^2(||\W_1\A^\plus||_F^2+||\W_3\A^\plus||_F^2)\\
&=&\error{\A}{\W_1}+\error{\A}{\W_3}\geq\error{\A}{\W_1}.
\end{eqnarray*}
Therefore, $\minerror(\W_1)\leq\minerror(\W_2)$.

For (ii), given a strategy $\A_2$ on $\W_2$, let $\A_1$ be a column projection of $\A_2$ using the same projection that generates $\W_1$ from $\W_2$. According to the construction of $\A_1$ and $\W_1$, since $\W_2=\W_2\A_2^\plus\A_2$,  we have $\W_2\A_2^\plus\A_1=\W_1$. Therefore according to Theorem~\ref{thm:mpinverse}, $||\W_2\A_2^\plus||_F\geq||\W_1\A_1^\plus||_F$. Furthermore, since $\sens{\A_2}\geq\sens{\A_1}$, we know $\error{\A_2}{\W_2}\geq\error{\A_1}{\W_1}$.
\end{proof}

\subsection{The Singular Value Bound}\label{app:proof:svdb}
We first present the proof of lemmas that are required for the proof of Thm~\ref{thm:singularvaluebound}. After that, we prove the theorems showing that the supreme SVD bound satisfies the error properties. 
\lemcolumnuniform*
\begin{proof}
For any workload $\W$, the problem of finding a strategy that minimizes the total error of $\W$ can be formulated as a SDP problem \cite{Li:2010Optimizing-Linear}. Therefore the optimal strategy that minimizes the total error of $\W$ always exists. 
Let $\A'$ be an optimal strategy on workload $\W$.  We now construct a matrix $\A$ from $\A'$ such that $\A\in\mathcal{A}_\W$.  Let $d_1, \ldots, d_n$  denote the diagonal entries of matrix $\sens{\A}^2\I - \A'^T\A'$, i.e. $d_1, \ldots, d_n$ is the difference between each diagonal entry of $\A'^T\A'$ and the maximal diagonal entry of $\A'^T\A'$.  Since $d_1,\ldots, d_n\geq 0$, let $\D$ be the diagonal matrix whose diagonal entries are $\sqrt{d_1},\ldots, \sqrt{d_n}$, and
$\A=\left[
\begin{smallmatrix}
\A'\\ \D\end{smallmatrix}
\right]$.
Then $\A$ is a strategy matrix such that the diagonal entries of $\A^T\A$ are all the same. Let $\B=[\A'^\plus, \mathbf{0}]$. Then $\W\B\A=\W\A'^\plus\A'=\W$.  According to Theorem~\ref{thm:mpinverse}, $||\W\A^\plus||_F\leq||\W\B||_F$. Recall $\sens{\A}=\sens{\A'}$, we have,
\begin{eqnarray*}
\error{\A}{\W} &=&P(\epsilon,\delta)\sens{\A}^2||\W\A^\plus||_F^2\\
&\leq&P(\epsilon,\delta)\sens{\A}^2||\W\B||_F^2\\
&=&P(\epsilon,\delta)\sens{\A'}^2||\W\A'^\plus||_F^2\\
&=&\error{\A'}{\W} = \minerror(\W).
\end{eqnarray*}
Therefore $\error{\A}{\W} =\minerror(\W)$ and $\A$ is an optimal strategy for workload $\W$.
\end{proof}

\lemdiagmin*
\begin{proof}
Use $d_i$ to denote the diagonal entries of $\D$ and $p_{ij}$ to denote the entries in $\P$. Noticing that $\sum_{j=1}^n p_{ji}^2=1$, we have
\[||\D\p_i||_2=\sqrt{\sum_{j=1}^n p_{ji}^2d_j^2}\geq \sum_{j=1}^n p_{ji}^2d_j.\]
Therefore, since $\sum_{j=1}^n p_{ij}^2=1$,
\[\sum_{i=1}^n||\D\p_i||_2\geq \sum_{i=1}^n\sum_{j=1}^n p_{ji}^2d_j=\sum_{j=1}^n(\sum_{i=1}^n p_{ji}^2)d_j=\tr(\D).\]
\end{proof}

\thmsvdbequiv*
\begin{proof}
(i) (ii) (iii): Since any one of those conditions leads to $\W_1^T\W_1=\W_2^T\W_2$, according to Prop.~\ref{prop:svdbequiv}, $\ssvdb(\W_1)=\ssvdb(\W_2)$.

(iv): Given a workload $\W_1$, it is sufficient to prove that the singular values of $\W_1$ are column representation independent. Let $\W_2$ be a matrix resulting from a permutation of the columns of $\W_1$ and $\P$ be the permutation matrix such that $\W_1\P=\W_2$. If a singular value decomposition of $\W_1$ is $\W_1=\Q_\W\LambdaB_\W\P_\W$, then the decomposition of $\W_2$ is $\W_2=\Q_\W\LambdaB_\W\P_\W\P$. Since $\P_\W\P$ is still an orthogonal matrix, $\W_2=\Q_\W\LambdaB_\W\P_\W\P$ is a singular value decomposition of $\W_2$. Therefore the singular values of $\W_2$ are exactly the same as the singular values of $\W_1$.

(v): Since $\W_2$ is a column projection of $\W_1$, $\ssvdb(\W_1)\geq\ssvdb(\W_2)$ by definition. In addition, for any matrix with columns of zeroes, removing thse columns will not impact the non-zero singular values of the matrix. Therefore projecting those columns out will reduce the total number of singular values but not their sum. Therefore projecting out all zero columns from $\W_1$ will not decrease $\ssvdb(\W_1)$, which indicates $\ssvdb(\W_1)=\ssvdb(\W_2)$.
\end{proof}

\thmsvdbinequal*
Since (ii) is naturally satisfied according to the definition of $\ssvdb(\W)$, it is sufficient to prove (i). Here we prove it is true even for the SVD bound so that it is also true for the supreme SVD bound. 
\begin{proof}
Given $\W_1\subseteq \W_2$, according to the definition, there exists a workload $\W'_2$ such that $\W'_2\equiv \W_2$ and $\W'_2$ contains all the queries of $\W_1$. Then $\W'_2$ has the following form:
\[\W'_2=\left[\begin{array}{c}
\W_1 \\
\W_3
\end{array}\right].\]
Then
\begin{eqnarray*}
\W_2^T\W_2-\W_1^T\W_1 &=& {\W'_2}^T\W'_2-\W_1^T\W_1\\
&=&\W_3^T\W_3\succeq 0
\end{eqnarray*}
Let $\W_1=\Q_1\LambdaB_1\P_1$ and $\W_2=\Q_2\LambdaB_2\P_2$ be the singular value decomposition of $\W_1$ and $\W_2$, respectively. Then
\begin{eqnarray*}
\W_2^T\W_2-\W_1^T\W_1 \succeq 0&\Leftrightarrow& \P_2^T\LambdaB_2^2\P_2-\P_1^T\LambdaB_1^2\P_1 \succeq 0\\
&\Leftrightarrow&\LambdaB_2^2-\P_2\P_1^T\LambdaB_1^2\P_1\P_2^T \succeq 0\\
&\Rightarrow&\forall\,i,\,\lambdaB_i\geq ||\LambdaB_1\p_i||_2,
\end{eqnarray*}
where $\lambdaB_i$ is the $i$-th diagonal entry of $\LambdaB_2$ and $\p_i$ is the $i$-th column vector of $\P_1\P_2^T$. The inequality in the last row based on the property that the diagonal entries of any positive semidefinite matrix are non-negative. Therefore, according to Lemma~\ref{lem:diagmin},
\begin{equation}\label{eqn:traceorder}
\tr(\LambdaB_2)=\sum_{i=1}^n\lambdaB_i\geq\sum_{i=1}^n  ||\LambdaB_1\p_i||_2\geq \tr(\LambdaB_1).
\end{equation}
\end{proof}

\subsection{Analysis of the SVD Bound}\label{app:proof:theory}
The theorems that are respectively related to the tightness and the looseness of the SVD bound are proved here.
\thmsvdtightness*
\begin{proof}
Recall that the SVD bound is tight if and only if (\ref{eqn:totracecond}), (\ref{eqn:cauchy}) and (\ref{eqn:minap}) takes equal sign simultaneously. The conditions that make all three inequalities to have equal sign is $\A\propto\Q\sqrt{\LambdaB_\W}\P_\W$ and $\A\in\mathcal{A}_\W$, which is equivalent to the case that the diagonal entries of  $\P_\W^T\LambdaB_\W\P_\W$ are all the same. In addition, $\P_\W^T\LambdaB_\W\P_\W=\sqrt{\W^T\W}$ and we have the theorem proved.
\end{proof}

\thmvaragn*
\begin{proof}
For the strategy $\A=\sqrt{\LambdaB_\W}\P_\W$, if $\A\in\mathcal{A}$, \\$\error{\A}{\W}=P(\epsilon,\delta)\svdb(\W)$. Since $\W$ is a variable agnostic matrix, $\W^T\W$ has the following special form:
\[\W^T\W=\small
\left[\begin{array}{cccc}
a & b & \ldots & b\\
b & a & \ldots & b\\
\vdots & \vdots & \ddots & \vdots \\
b & b & \ldots & a
\end{array}\right],
\]
where $a>b$. One can verify that $a+(n-1)b$ is an eigenvalue of $\W^T\W$ with order 1 and $a-b$ is an eigenvalue of $\W^T\W$ with order $n-1$. Let $\P_\W=(p_{ij})_{n\times n}$. Noticing the rows of $\P_\W$ are the eigenvectors of $\W^T\W$, without loss of generality, assume its first row contains the eigenvector associated with eigenvalue $a+(n-1)b$. In addition, since $\W^T\W\vect{1}=(a+(n-1)b)\vect{1}$, $(1/\sqrt{a+(n-1)b})\vect{1}$ is an unit-length eigenvector of $\W^T\W$ associated with eigenvalue $a+(n-1)b$. The $i^{th}$ diagonal entry of matrix $\A^T\A=\P_\W^T\LambdaB_\W\P_\W$ can be computed as:
\begin{align*}
(a+(n-1)b)p_{1i}^2+\sum_{j=2}^n (a-b)p_{ji}^2&=1+(a-b)(\sum_{j=1}^n p_{ji}^2 - p_{1i}^2)\\
&=1+\frac{a-b}{a+(n-1)b}.
\end{align*}
Thus $\A\in\mathcal{A}$ and so it is a strategy for $\W$ such that \\
$\error{\A}{\W}=P(\epsilon,\delta)\svdb(\W)$. Therefore the SVD bound is achieved.
\end{proof}

\thmdatacube*
\begin{proof}
Let us induct on the number of attributes $d$ in the database. When $d=1$, there are only two cuboids, the cuboid asks for the sum of all the cells and the cuboid asks for all the individual cells. Consider the workload $\W$ that weight the first cuboid $w_1$ and the second cuboid $w_2$, one can compute that
\[\W^T\W=\small
\left[\begin{array}{cccc}
w_1^2+w_2^2 & w_1^2 & \ldots & w_1^2\\
w_1^2 & w_1^2+w_2^2 & \ldots & w_1^2\\
\vdots & \vdots & \ddots & \vdots \\
w_1^2 & w_1^2 & \ldots & w_1^2+w_2^2
\end{array}\right],
\]
which is a variable agnostic workload. Therefore, according to Thm~\ref{thm:varagn}, $\svdb(\W)$ is tight.

If the SVD bound is tight when $d=d_0$, consider the case that $d=d_0+1$. Given a data cube workload $\W$. The cuboids in the data cube can be separated into two groups: the first group is the cuboids that aggregate on the last attribute; the second group is the cuboids that do not aggregate on the last attribute. Let $\W_1$ be the projection of the cuboids in the first group on the first $d_0$ attributes and $\W_2$ be the projection of the cuboids in the first group on the first $d_0$ attributes. We can represent $\W$ using $\W_1$ and $\W_2$:
\[\W=\small\begin{bmatrix}\W_1 & \W_1 & \ldots &  \W_1 \\
\W_2 & 0 & \ldots & 0 \\
0 & \W_2 & \ldots & 0 \\
\vdots & \vdots & \ddots & \vdots \\
0 & 0 & \ldots & \W_2
\end{bmatrix},\]
where the number of $\W_1$ blocks and $\W_2$ blocks are the number of values in the last attribute, denoted as $n_0$. Let $\Q_1$, $\Q_2$ be the orthogonal matrices such that $\Q_1\W_1=\sqrt{\W_1^T\W_1}$ and $\Q_2\W_1=\sqrt{\W_2^T\W_2}$. Let
\[\Q=\small\begin{bmatrix}\Q_1 & 0 & \ldots & 0 \\
0 & \Q_2 & \ldots & 0 \\
\vdots & \vdots & \ddots & 0 \\
0 & 0 & \ldots & \Q_2
\end{bmatrix},\]
and then
\[\Q\W=\small\begin{bmatrix}\sqrt{\W_1^T\W_1} & \sqrt{\W_1^T\W_1} & \ldots &  \sqrt{\W_1^T\W_1} \\
\sqrt{\W_2^T\W_2} & 0 & \ldots & 0 \\
0 & \sqrt{\W_2^T\W_2} & \ldots & 0 \\
\vdots & \vdots & \ddots & \vdots \\
0 & 0 & \ldots & \sqrt{\W_2^T\W_2}
\end{bmatrix}.\]

One can verify that
\[\sqrt{\W^T\W}=\sqrt{\W^T\Q^T\Q\W}=\small\begin{bmatrix}\W_3 & \W_4 & \ldots &  \W_4 \\
\W_4 & \W_3 & \ldots & \W_4 \\
\vdots & \vdots & \ddots & \vdots \\
\W_4 & \W_4 & \ldots & \W_3
\end{bmatrix},\]
where
\begin{align*}
\W_3&=\frac{1}{n_0}((n_0-1)\sqrt{\W_1^T\W_1}+\sqrt{\W_1^T\W_1+n_0\W_2^T\W_2}),\\
\W_4&=\frac{1}{n_0}(-\sqrt{\W_1^T\W_1}+\sqrt{\W_1^T\W_1+n_0\W_2^T\W_2}).
\end{align*}
Noticing that both $\W_1$ and $\left[\begin{smallmatrix}\W_1\\ \sqrt{n_0}\W_2\end{smallmatrix}\right]$ are data cube workloads on $d_0$ attributes, according to the induction assumptions, both $\sqrt{\W_1^T\W_1}$ and $\sqrt{\W_1^T\W_1+n_0\W_2^T\W_2}$ are symmetric matrices whose diagonal entries are all the same, respectively. Thus $\W_3$ and $\W_4$ are also symmetric matrices whose diagonal entries are all the same, respectively. Then  $\sqrt{\W^T\W}$ is a symmetric matrices whose diagonal entries are all the same and then the SVD bound is tight on $\W$.
\end{proof}


\thmsvdratio*
\begin{proof}
\begin{align*}
\minerror(\W)&\leq\error{\A}{\W}\\
&=\frac{nd_0P(\epsilon, \delta)\svdb(\W)}{\trace(\sqrt{\W^T\W})}.
\end{align*}
\end{proof}

\subsection{Asymptotic Estimation to the SVD Bound}\label{app:proof:comparison}
Here we prove the our asymptotic estimation to the SVD bound.
\thmasysvd*
\begin{proof}
Given a workload $\W$. Let $\lambda_1, \ldots, \lambda_{\min(m,n)}$ be the non-zero singular values of $\W$.
\begin{align*}
(\lambda_1+\ldots+\lambda_{\min(m,n)})^2&\leq\min(m,n)(\lambda_1^2+\ldots+\lambda_{\min(m,n)}^2)\\
&=\min(m,n)||\W||_F^2\\
&\leq\min(m,n)mn
\end{align*}
Therefore
\[\svdb(\W)\leq \min(m,n)m.\]
Noticing $\svdb(\W)$ is estimating the $L_2$ error of $m$ queries and the error is Gaussian random noise. Take the average
of the $\svdb(\W)$, consider the error estimator for Gaussian random variable with mean $m$ and standard deviation $\sigma$:
\[\P(|X-m|>t\sigma)\leq\frac{\sqrt{2}}{\sqrt{\pi}t}\exp(-\frac{t^2}{2}),\]
and we have the bound proved.
\end{proof}
\subsection{The SVD Bound and Operations on Workloads}\label{app:proof:operation}
Here we prove relations between the SVD bound and the operations on workload.  The proofs of union and generalized negation are related to the relationship between singular values of matrices and their sum, as stated in the theorem below.

\begin{proposition}[\cite{fulton2000eigenvalues}]\label{thm:sumofsv}
Given two $n\times n$ matrices $\W_1$ and $\W_2$ with singular values $\mu_1, \mu_2, \ldots, \mu_n$ and $\lambda_1, \lambda_2, \ldots, \lambda_n$ respectively. Let $\phi_1, \phi_2, \ldots, \phi_n$ be the singular values of $\W_1+\W_2$, then
\[ \sum_{i=1}^n \mu_i + \sum_{i=1}^n \lambda_i \geq \sum_{i=1}^n \phi_i.\]
\end{proposition}

For any $m\times n$ matrix $\W$, there always exists an $n\times n$ matrix $\W'$ such that the nonzero singular values of $\W$ and $\W'$ are all the same. Theorem~\ref{thm:sumofsv} holds even if both $\W_1$ and $\W_2$ are $m\times n$ matrices, which leads to the relationship between the SVD bounds of two workloads and their sum.

\thmsvdofsum*
\begin{proof}
Let $\W=\W_1\cup\W_2$. Expand $\W_1$, $\W_2$ to two $(m_1+m_2)\times n$ matrices as follows:
\[\W_1'=\left[\begin{array}{c}\W_1 \\ 0 \end{array}\right],
\quad \W_2'=\left[\begin{array}{c}0 \\ \W_2 \end{array}\right].\]
Since $\W_1'$ and $\W_2'$ have the same singular values as $\W_1$ and $\W_2$, respectively, $\svdb(\W_1')=\svdb(\W_1)$, $\svdb(\W_2')=\svdb(\W_2)$. Furthermore, since
\[\W_1'+\W_2'=\left[\begin{array}{c}\W_1 \\ \W_2 \end{array}\right]\supseteq\W, \]
according to Prop.~\ref{thm:sumofsv}, the sum of the singular values of $\W_1'$ and $\W_2'$ is larger than or equal to the sum of singular values $\W$. Therefore, with Thm.~\ref{thm:svdbinequal} and Prop.~\ref{thm:sumofsv}, 
\begin{eqnarray*}
&&\sqrt{\svdb(\W_1)}+\sqrt{\svdb(\W_2)}\\
&=&\sqrt{\svdb(\W_1')}+\sqrt{\svdb(\W_2')}\geq\sqrt{\svdb(\W)}.
\end{eqnarray*}
For the case of $\ssvdb$, notice that for any projection $\mu$, 
\[\sqrt{\svdb(\mu(\W_1))}+\sqrt{\svdb(\mu(\W_2))}\geq\sqrt{\mu(\svdb(\W))}.\]
Consider all projections and we have the result proved.
\end{proof}


The property of crossproduct can be proved by constructing a proper representation to the resulting workload.
\thmsvdofxprod*
\begin{proof}

Let $\w_1$ and $\w_2$ be queries in $\W_1$ and $\W_2$, respectively and $\W=\W_1\times \W_2$. Consider the vector representation of $\w_1$ and $\w_2$: $\w_1=[w_{11}, w_{12}, \ldots, w_{1n_1}]^T$, $\w_2=[w_{21}, w_{22}, \ldots, w_{2n_2}]^T$.
The crossproduct of $\w_1$ and $\w_2$, denoted as $\w$ can be represented as an $n_1$ by $n_2$ matrix, whose $(i,j)$ entry is equal to $w_{1i}w_{2j}$. In another word,
\[\w=\w_1\w_2^T.\]
We can the represent $\w$ as a vector, denoted as $\w'$, which is a $1\times n_1n_2$ vector that contains entries in $\w$ row by row. Therefore,
\[\w'=[w_{11}\w_2, w_{12}\w_2, \ldots, w_{1n}\w_2]^T.\]
More generally, using $\W^1_{ij}$ to denote the $(i,j)$ entry in $\W_1$, $\W$ can be represented as the following matrix:
\[\W=\left[\begin{array}{cccc}
w^1_{11}\W_2 & w^1_{12}\W_2 & \ldots & w^1_{1n_1}\W_2 \\
w^1_{21}\W_2 & w^1_{22}\W_2 & \ldots & w^1_{2n_1}\W_2 \\
\vdots & \vdots & \ddots & \vdots\\
w^1_{n1}\W_2 & w^1_{n2}\W_2 & \ldots & w^1_{m_1n_1}\W_2
\end{array}\right].\]
Let $\v_1$, $\v_2$ be the eigenvectors of $\W_1^T\W_1$, $\W_2^T\W_2$ with eigenvalues $\lambda_1$, $\lambda_2$, respectively. Let the vector representation of $\v_1$ be $\v_1=[v_{1n}, v_{2n}, \ldots, v_{1n}]^T$. Consider the following vector
\[\v=[v_{11}\v_2, v_{12}\v_2, \ldots, v_{1n}\v_2]^T,\]
According to block matrix multiplication,
\begin{eqnarray*}
\W^T\W\v
&=& \W^T\left[\begin{array}{c}
\sum_{i=1}^n w^1_{1i}v_{1i}\W_2\v_2 \\ \sum_{i=1}^n w^1_{2i}v_{1i}\W_2\v_2 \\
 \vdots \\ \sum_{i=1}^n w^1_{m_1i}v_{1i}\W_2\v_2
\end{array}\right] 
\end{eqnarray*}
\begin{eqnarray*}
&=& \left[\begin{array}{c}
\lambda_1v_{11}\lambda_2\v_2 \\ \lambda_1v_{12}\lambda_2\v_2 \\ \vdots \\ \lambda_1v_{1n}\lambda_2\v_2
\end{array}\right] = \lambda_1\lambda_2\v.
\end{eqnarray*}
Thus $\v$ is an eigenvector of $\W^T\W$ with eigenvalue $\lambda_1\lambda_2$. Since $\W_1^T\W_1$ and $\W_2^T\W$ have $n_1$ and $n_2$ orthogonal eigenvectors, respectively, we can find $n_1n_2$ orthogonal eigenvectors with this method. Noticing $\W^T\W$ only has $n_1n_2$ eigenvalues, the eigenvalues of those $n_1n_2$ eigenvectors are all the eigenvalues of $\W^T\W$. Let $\lambda_{11}, \ldots, \lambda_{1n}$ be the eigenvalues of $\W_1^T\W_1$ and $\lambda_{21}, \ldots, \lambda_{2n}$ be the eigenvalues of $\W_2^T\W_2$.
\begin{eqnarray*}
\svdb(\W)&=&\frac{1}{n_1n_2}(\sum_{1\leq i\leq n_1, 1\leq j\leq n_2}\sqrt{\lambda_{1i}\lambda_{2j}})^2\\
&=&\frac{1}{n_1}(\sum_{i=1}^{n_1}\sqrt{\lambda_{1i}})^2\cdot \frac{1}{n_2}(\sum_{i=1}^{n_2}\sqrt{\lambda_{2i}})^2 \\
&=&\svdb(\W_1)\svdb(\W_2).
\end{eqnarray*}
Though we use a specific rule above to represent the query cross products as query vectors, according to Theorem~\ref{thm:svdbequiv}, the SVD bound is independent of the rule of representation. Thus we have the theorem proved in arbitrary cases. 

For the case of $\ssvdb$, since for any projection $\mu_1$ on $\W_1$ and $\mu_2$ on $\W_2$, $\mu_1\times\mu_2$ is a projection on $\W$. On the another hand, there are projections on $\W$ that can not be represented as a crossproduct of a projection on $\W_1$ and a projection on $\W_2$. Therefore
\begin{align*}
\ssvdb(\W_1)\ssvdb(\W_2)&=\max_{\mu_1}\svdb(\mu_1(\W_1))\max_{\mu_2}\svdb(\mu_2(\W_2))\\
&=\max_{\mu_1, \mu_2}\svdb((\mu_1\times\mu_2)(\W))\\
&\leq\max_{\mu}\svdb(\mu(\W)) = \ssvdb(\W).
\end{align*}
\end{proof}

\end{document}